\documentclass[a4paper,UKenglish]{lipics-v2018}

\usepackage{microtype,boxedminipage,comment}

\title{Finding~a~Small~Number~of~Colourful~Components\footnote{This paper was supported by Research Project Grant RPG-2016-258 of the Leverhulme Trust.}}
\titlerunning{Finding a Small Number of Colourful Components}

\author{Laurent Bulteau}{Universit\'e Paris-Est, LIGM (UMR 8049), CNRS, ENPC, UPEM, ESIEE Paris, France}{laurent.bulteau@u-pem.fr}{}{}
\author{Konrad K. Dabrowski}{Department of Computer Science, Durham University, Durham, UK}{konrad.dabrowski@durham.ac.uk}{}%{}
{Also supported by EPSRC (EP/K025090/1)}% and the Leverhulme Trust (RPG-2016-258).}
\author{Guillaume Fertin}{Universit\'e de Nantes, LS2N (UMR 6004), CNRS, Nantes, France}{guillaume.fertin@univ-nantes.fr}{}{}
\author{Matthew Johnson}{Department of Computer Science, Durham University, Durham, UK}{matthew.johnson2@durham.ac.uk}{}{}
\author{Dani\"el Paulusma}{Department of Computer Science, Durham University, Durham, UK}{daniel.paulusma@durham.ac.uk}{}%{}
{Also supported by EPSRC (EP/K025090/1)}% and the Leverhulme Trust (RPG-2016-258).}
\author{St\'ephane Vialette}{Universit\'e Paris-Est, LIGM (UMR 8049), CNRS, ENPC, UPEM, ESIEE Paris, France}{vialette@u-pem.fr}{}{}
\authorrunning{L. Bulteau, K.\,K. Dabrowski, G. Fertin, M. Johnson, D. Paulusma and S. Vialette}

\Copyright{Laurent Bulteau, Konrad K. Dabrowski, Guillaume Fertin, Matthew Johnson, Dani\"el Paulusma and St\'ephane Vialette}

\subjclass{}


\EventEditors{John Q. Open and Joan R. Access}
\EventNoEds{2}
\EventLongTitle{42nd Conference on Very Important Topics (CVIT 2016)}
\EventShortTitle{CVIT 2016}
\EventAcronym{CVIT}
\EventYear{2016}
\EventDate{December 24--27, 2016}
\EventLocation{Little Whinging, United Kingdom}
\EventLogo{}
\SeriesVolume{42}
\ArticleNo{23}
\nolinenumbers %uncomment to disable line numbering
\hideLIPIcs  %uncomment to remove references to LIPIcs series (logo, DOI, ...), e.g. when preparing a pre-final version to be uploaded to arXiv or another public repository; 
\usepackage{tikz}
\usetikzlibrary{shapes,calc}
\usetikzlibrary{backgrounds}

\newcounter{ctrclaim}[theorem]
\newcounter{ctrcase}[theorem]
\newcounter{ctrrule}[theorem]
\newcommand{\clm}[1]{\medskip\phantomsection\refstepcounter{ctrclaim}\noindent\textcolor{darkgray}{$\blacktriangleright$\nobreakspace\sffamily\bfseries Claim \thectrclaim. }{\em #1}}
\newcommand{\case}[1]{\medskip\phantomsection\refstepcounter{ctrcase}\noindent\textcolor{darkgray}{$\blacktriangleright$\nobreakspace\sffamily\bfseries Case \thectrcase. }{\em #1}}
\newcommand{\ourrule}[1]{\medskip\phantomsection\refstepcounter{ctrrule}\noindent\textcolor{darkgray}{$\blacktriangleright$\nobreakspace\sffamily\bfseries Rule \thectrrule. }{\em #1}}

\newcommand{\True}{\ensuremath{\text{\tt true}}}
\newcommand{\False}{\ensuremath{\text{\tt false}}}

\newcommand \dia{\hfill{\textcolor{darkgray}{$\diamond$}}}

\usetikzlibrary{backgrounds}
\usetikzlibrary{decorations.pathreplacing}

\tikzset{every node/.style={circle, inner sep=0pt, minimum size=3pt, draw=black, fill=white}}

\tikzset{every path/.style={draw=black}}
\definecolor{colour1}{rgb}{.8,0,0}
\definecolor{colour2}{rgb}{0.7,0.7,1}
\tikzset{comp1/.style = {every path/.style={
			draw=colour1,
			line width=3pt
}}}
\tikzset{comp2/.style = {every path/.style={
			draw=colour2,
			line width=3pt
}}}

\tikzset{maybe/.style = {dashed }}
\newcommand{\vertexGadgetThreeComponents}{
\begin{scope}[scale=.6, label distance=1mm]
\foreach \d/\l/\a/\ca/\cb/\cc in {150/60/a/1/3/2, 270/180/b/3/2/1,  30/300/c/2/1/3 } {
 \path [] (\d:.7) node (\a1)[label=\l:\ca]{}
        -- ++ (\d:1)node (\a2)[label=\l:\cb]{}
        -- ++ (\d:1)node (\a3)[label=\l:\cc]{}
        -- ++ (\d:1)node (\a4)[label=\l:\ca, label={[label distance=8mm]\l:$v_{e_\ca}$}]{}  ;
 \path (\a4) -- ++ (\d:0.6);
 \path [dotted] (\d:4.3) -- ++ (\d:0.6) node[draw=none] {};

 \begin{scope}[on background layer]
 	\draw[draw=none, fill=red!40!white, rotate=\l] (\a4) ellipse (1 and .5);
 \end{scope}
 }
 \path[rounded corners] (a1) --(b1)--(c1)--(a1) ;
 \begin{scope}[on background layer]
 	\path [draw=none, fill=blue!20!white, on background layer] (0,0) circle (3);
 \end{scope}
 \end{scope}
 }

\newcommand{\vargadgetPathwidth}{

 \path [] (0,0) node (a)[label=above:$a_i$, label=left:$\cdots$]{}
 	   -- ++(1,0) node (b)  [label=above:$b_i$] {}
 	   -- ++(1,0) node (c) [label=above:$c_i$] {}
 	   -- ++(5,0) node (d) [label=above:$d_i$,label=right:$\cdots$] {};

\path  (c)
 	   -- ++(1,-1) node (x1) [label=below:$\alpha_i^1$] {}
 	   -- ++(1,0) node  [label=below:$\alpha_i^2$] {}
 	   -- ++(1,0) node  [label=below:$\cdots$] {}
 	   -- ++(1,0) node  (xl) [label=below:$\alpha_i^{2\ell_i}$] {};

\path (0,-2) node (a2)[label=below:$a_i$, label=left:$\cdots$]{}
 	   -- ++(1,0) node (b2)  [label=below:$b_i$] {}
 	   -- ++(1,0) node (c2) [label=below:$c_i$] {}
 	   -- ++(5,0) node (d2) [label=below:$d_i$, label=right:$\cdots$] {};

 \path (a)--(b2);
 \path (b)--(a2);
 }
\newcommand{\vargadgetPathwidthnolines}{

 \path [] (0,0) node (a)[label=above:$a_i$, label=left:$\cdots$]{}
 	   ++(1,0) node (b)  [label=above:$b_i$] {}
 	   ++(1,0) node (c) [label=above:$c_i$] {}
 	   ++(5,0) node (d) [label=above:$d_i$,label=right:$\cdots$] {};

\path  (c)
 	   ++(1,-1) node (x1) [label=below:$\alpha_i^1$] {}
 	   ++(1,0) node  [label=below:$\alpha_i^2$] {}
 	   ++(1,0) node  [label=below:$\cdots$] {}
 	   ++(1,0) node  (xl) [label=below:$\alpha_i^{2\ell_i}$] {};

\path (0,-2) node (a2)[label=below:$a_i$, label=left:$\cdots$]{}
 	   ++(1,0) node (b2)  [label=below:$b_i$] {}
 	   ++(1,0) node (c2) [label=below:$c_i$] {}
 	   ++(5,0) node (d2) [label=below:$d_i$, label=right:$\cdots$] {};

 }
\newcommand{\vargadgetPathwidthnonodes}{

 \path [] (0,0) node[minimum size=0pt] (Za){}
 	   ++(1,0) node[minimum size=0pt] (Zb)   {}
 	   ++(1,0) node[minimum size=0pt] (Zc)  {}
 	   ++(5,0) node[minimum size=0pt] (Zd)  {};

\path  (Zc)
 	   ++(1,-1) node[minimum size=0pt] (Zx1)  {}
 	   ++(1,0) node[minimum size=0pt]   {}
 	   ++(1,0) node[minimum size=0pt]   {}
 	   ++(1,0) node[minimum size=0pt]  (Zxl)  {};

\path (0,-2) node[minimum size=0pt] (Za2){}
 	   ++(1,0) node[minimum size=0pt] (Zb2)   {}
 	   ++(1,0) node[minimum size=0pt] (Zc2)  {}
 	   ++(5,0) node[minimum size=0pt] (Zd2)  {};

 \path (Za) -- (Zb2) -- (Za2) -- (Zb) -- (Za);
 }
\newcommand{\clausegadgetPathwidth}{

 \path (0,0) node (a) [label=above:$e_j$, label=left:$\cdots$]{}
 	   -- ++(1,0) node (b)  [label=above:$f_j$] {}
 	   -- ++(1.5,0) node (c) [label=above:$g_j$] {}
 	   -- ++(1.5,0) node (f) [label=above:$h_j$] {}
 	   -- ++(1.5,0) node (g) [label=above:$i_j$, label=right:$\cdots$] {};

  \path (0,-2) node (a2) [label=below:$e_j$, label=left:$\cdots$]{}
 	   -- ++(1,0) node (b2)  [label=below:$f_j$] {}
 	   -- ++(1.5,0) node (c2) [label=below:$g_j$] {}
 	   -- ++(1.5,0) node (f2) [label=below:$h_j$] {}
 	   -- ++(1.5,0) node (g2) [label=below:$i_j$, label=right:$\cdots$] {};
 \path (c)
       -- ++ (0,-.5) node (x1) [label=right:$\alpha_g^{2r-1}$] {}
       -- ++ (0,-.5) node (C1) [label=right:$\beta_j$] {}
       -- ++ (0,-.5) node (y) [label=right:$\alpha_h^{2s}$] {}
       -- (c2);
 \path (a)--(b2);
 \path (b)--(a2);
 \path(f)
       -- ++ (0,-.5) node (x2) [label=right:$\alpha_g^{2r}$] {}
       -- ++ (0,-.5) node (C2) [label=right:$\beta_j$] {}
       -- ++ (0,-.5) node (z) [label=right:$\alpha_i^{2t}$] {}
       -- (f2);
}

\newcommand{\clausegadgetPathwidthnolines}{

 \path (0,0) node (a) [label=above:$e_j$, label=left:$\cdots$]{}
 	   ++(1,0) node (b)  [label=above:$f_j$] {}
 	   ++(1.5,0) node (c) [label=above:$g_j$] {}
 	   ++(1.5,0) node (f) [label=above:$h_j$] {}
 	   ++(1.5,0) node (g) [label=above:$i_j$, label=right:$\cdots$] {};

  \path (0,-2) node (a2) [label=below:$e_j$, label=left:$\cdots$]{}
 	   ++(1,0) node (b2)  [label=below:$f_j$] {}
 	   ++(1.5,0) node (c2) [label=below:$g_j$] {}
 	   ++(1.5,0) node (f2) [label=below:$h_j$] {}
 	   ++(1.5,0) node (g2) [label=below:$i_j$, label=right:$\cdots$] {};
 \path (c)
       ++ (0,-.5) node (x1) [label=right:$\alpha_g^{2r-1}$] {}
       ++ (0,-.5) node (C1) [label=right:$\beta_j$] {}
       ++ (0,-.5) node (y) [label=right:$\alpha_h^{2s}$] {}
       (c2);
 \path (a) (b2);
 \path (b) (a2);
 \path(f)
       ++ (0,-.5) node (x2) [label=right:$\alpha_g^{2r}$] {}
       ++ (0,-.5) node (C2) [label=right:$\beta_j$] {}
       ++ (0,-.5) node (z) [label=right:$\alpha_i^{2t}$] {}
       (f2);
}

\newcommand{\clausegadgetPathwidthnonodes}{

 \path (0,0) node[minimum size=0pt] (Za) {}
 	   ++(1,0) node[minimum size=0pt] (Zb)   {}
 	   ++(1.5,0) node[minimum size=0pt] (Zc)  {}
 	   ++(1.5,0) node[minimum size=0pt] (Zf)  {}
 	   ++(1.5,0) node[minimum size=0pt] (Zg)  {};

  \path (0,-2) node[minimum size=0pt] (Za2) {}
 	   ++(1,0) node[minimum size=0pt] (Zb2)   {}
 	   ++(1.5,0) node[minimum size=0pt] (Zc2)  {}
 	   ++(1.5,0) node[minimum size=0pt] (Zf2)  {}
 	   ++(1.5,0) node[minimum size=0pt] (Zg2)  {};
 \path (Zc)
       ++ (0,-.5) node[minimum size=0pt] (Zx1)  {}
       ++ (0,-.5) node[minimum size=0pt] (ZC1)  {}
       ++ (0,-.5) node[minimum size=0pt] (Zy)  {}
       (Zc2);
 \path (Zc) -- (Zc2);
 \path (Za) -- (Zb2) -- (Za2) -- (Zb) -- (Za);
 \path(Zf)
       ++ (0,-.5) node[minimum size=0pt] (Zx2)  {}
       ++ (0,-.5) node[minimum size=0pt] (ZC2)  {}
       ++ (0,-.5) node[minimum size=0pt] (Zz)  {}
       (Zf2);
 \path (Zf) -- (Zf2);
}

\theoremstyle{plain}

\theoremstyle{definition}

\newcommand{\ol}[1]{\overline{#1}}
\newcommand{\NP}{{\sf NP}}
\newcommand{\FPT}{{\sf FPT}}

\DeclareMathOperator{\tw}{tw}
\DeclareMathOperator{\vc}{vc}

\newcommand{\problemdef}[3]{
	\begin{center}
		\begin{boxedminipage}{.99\textwidth}
			\textsc{{#1}}\\[2pt]
			\begin{tabular}{ r p{0.8\textwidth}}
				\textit{~~~~Instance:} & {#2}\\
				\textit{Question:} & {#3}
			\end{tabular}
		\end{boxedminipage}
	\end{center}
}

\begin{document}

\maketitle

\begin{abstract}
A partition $(V_1,\ldots,V_k)$ of the vertex set of a graph~$G$ with a (not necessarily proper) colouring~$c$ is colourful if no two vertices in any~$V_i$ have the same colour and every set~$V_i$ induces a connected graph.
The {\sc Colourful Partition} problem is to decide whether a coloured graph $(G,c)$ has a colourful partition of size at most~$k$.
This problem is closely related to the {\sc Colourful Components} problem, which is to decide whether a graph can be modified into a graph whose connected components form a colourful partition by deleting at most~$p$ edges.
Nevertheless we show that {\sc Colourful Partition} and {\sc Colourful Components} may have different complexities for restricted instances.
We tighten known \NP-hardness results for both problems and in addition we prove new hardness and tractability results for {\sc Colourful Partition}.
Using these results we complete our paper with a thorough parameterized study of {\sc Colourful Partition}.
\end{abstract}

\section{Introduction}\label{s-intro}

Research in comparative genomics, which studies the structure and evolution of genomes from different species, has motivated a number of interesting graph colouring problems.
In this paper we focus on the multiple genome alignment problem, where one takes a set of sequenced genomes, lets the genes be the vertex set of a graph~$G$, and joins by an edge any pair of genes whose \emph{similarity} (determined by their nucleotide sequences) exceeds a given threshold.
The vertices are also coloured to indicate the species to which each gene belongs.
This leads to a {\em coloured graph} $(G,c)$, where $c:V(G)\to \{1,2,\ldots\}$ denotes the colouring of~$G$.
We emphasize that~$c$ is not necessarily proper (adjacent vertices may have the same colour).

One seeks to better understand the evolutionary processes affecting these genomes by attempting to partition the genes into \emph{orthologous} sets (that is, collections of genes that originated from the same ancestral species but diverged following a speciation event) that are sufficiently similar.
This translates into partitioning~$V(G)$ such that each part
\begin{enumerate}[(i)]
\item contains no more than one vertex of each colour, and
\item induces a connected component.
\end{enumerate}
These two conditions ensure that each part contains vertices representing an orthologous set of similar genes.
In addition, one seeks to find a partition that is in some sense optimal.

When Zheng et al.~\cite{ZSLS11} considered this model, one approach they followed was to try to delete as few edges as possible such that the connected components of the resulting graph~$G'$ are \emph{colourful}, that is, contain no more than one vertex of any colour; in this case the connected components give the partition of~$V(G)$.
This led them to the {\sc Orthogonal Partition} problem, introduced in~\cite{HLZ00} and also known as the {\sc Colourful Components} problem~\cite{AP15,BHKNTU12,DS18}. (Note that a graph is {\em colourful} if each of its components is colourful.)

\problemdef{Colourful Components}{A coloured graph~$(G,c)$ and an integer $p\geq 1$}{Is it possible to modify~$G$ into a colourful graph~$G'$ by deleting at most~\mbox{$p$ edges}?}
The central problem of this paper is a natural variant, {\sc Colourful Partition}, which was introduced by Adamaszek and Popa~\cite{AP15}.
A partition is {\em colourful} if every partition class induces a connected colourful graph.
The {\em size} of a partition is its number of partition classes.

\problemdef{Colourful Partition}{A coloured graph~$(G,c)$ and an integer $k\geq 1$}{Does~$(G,c)$ have a colourful partition of size at most~$k$?}
The following example demonstrates that {\sc Colourful Components} and {\sc Colourful Partition} are not the same:
a colourful partition with the fewest parts might require the deletion of more edges than the minimum needed to obtain a colourful graph.
Later we find a family of instances on which the problems have different complexities (see Corollary~\ref{c-diff}).

\begin{example}\label{e-intro}
Let~$G$ have vertices $u_1,\ldots,u_k$, $v_1,\ldots,v_k$, $w,w'$ and edges~$ww'$ and $u_iw$, $u_iw'$, $v_iw$, $v_iw'$ for $i\in\{1,\ldots,k\}$.
Let~$c$ assign colour~$i$ to each~$u_i$ and~$v_i$, colour~$k+\nobreak 1$ to~$w$ and colour~$k+\nobreak 2$ to~$w'$.
Then $(\{u_1,\ldots,u_k,w\},\{v_1,\ldots,v_k,w'\})$ is a colourful partition for~$(G,c)$ of size~$2$, which we obtain by deleting~$2k+\nobreak 1$ edges.
However, deleting the~$2k$ edges~$v_iw$ and~$v_iw'$ for $i\in\{1,\ldots,k\}$ also yields a colourful graph
(with a colourful partition of size~$k+1$).
See \figurename~\ref{fig:example} for an illustration.
\end{example}

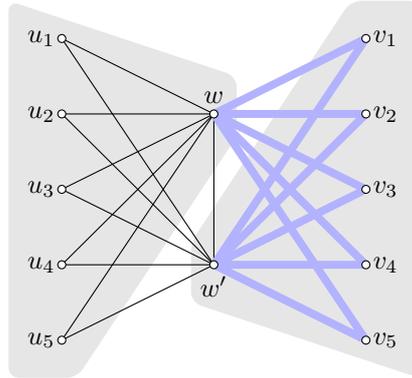
\begin{figure}
\begin{center}
{\begin{tikzpicture}

\node[minimum size=0pt] (Zw) at (0,1) {};
\node[minimum size=0pt] (Zwa) at (0,-1) {};

\node[minimum size=0pt] (Zu1) at (-2,2) {};
\node[minimum size=0pt] (Zu2) at (-2,1) {};
\node[minimum size=0pt] (Zu3) at (-2,0) {};
\node[minimum size=0pt] (Zu4) at (-2,-1) {};
\node[minimum size=0pt] (Zu5) at (-2,-2) {};

\node[minimum size=0pt] (Zv1) at (2,2) {};
\node[minimum size=0pt] (Zv2) at (2,1) {};
\node[minimum size=0pt] (Zv3) at (2,0) {};
\node[minimum size=0pt] (Zv4) at (2,-1) {};
\node[minimum size=0pt] (Zv5) at (2,-2) {};

\path (Zu1) -- (Zw) -- (Zv1) -- (Zwa) -- (Zu1); 
\path (Zu2) -- (Zw) -- (Zv2) -- (Zwa) -- (Zu2); 
\path (Zu3) -- (Zw) -- (Zv3) -- (Zwa) -- (Zu3); 
\path (Zu4) -- (Zw) -- (Zv4) -- (Zwa) -- (Zu4); 
\path (Zu5) -- (Zw) -- (Zv5) -- (Zwa) -- (Zu5); 

\path (Zw) -- (Zwa);

	   \begin{scope}[comp2]
	   \path (Zw)--(Zv1)--(Zwa);
	   \path (Zw)--(Zv2)--(Zwa);
	   \path (Zw)--(Zv3)--(Zwa);
	   \path (Zw)--(Zv4)--(Zwa);
	   \path (Zw)--(Zv5)--(Zwa);
	   \end{scope}

\node[label=above:$w$] (w) at (0,1) {};
\node[label=below:$w'$] (wa) at (0,-1) {};

\node[label=left:$u_1$] (u1) at (-2,2) {};
\node[label=left:$u_2$] (u2) at (-2,1) {};
\node[label=left:$u_3$] (u3) at (-2,0) {};
\node[label=left:$u_4$] (u4) at (-2,-1) {};
\node[label=left:$u_5$] (u5) at (-2,-2) {};

\node[label=right:$v_1$] (v1) at (2,2) {};
\node[label=right:$v_2$] (v2) at (2,1) {};
\node[label=right:$v_3$] (v3) at (2,0) {};
\node[label=right:$v_4$] (v4) at (2,-1) {};
\node[label=right:$v_5$] (v5) at (2,-2) {};
 
 \begin{scope}[on background layer]
 	\path [draw=none, rounded corners, fill=black!10!white, on background layer] (-2.7,-2.5) -- (-2.7,2.5) -- (0.3,1.5) -- (0.3,0.7) -- (-1.8,-2.5) -- cycle;
 \end{scope}

 \begin{scope}[on background layer]
 	\path [draw=none, rounded corners, fill=black!10!white, on background layer] (2.7,2.5) -- (2.7,-2.5) -- (-0.3,-1.5) -- (-0.3,-0.7) -- (1.8,2.5) -- cycle;
 \end{scope}

\end{tikzpicture}}
\end{center}
\caption{\label{fig:example}The graph~$G$ of Example~\ref{e-intro} with $k=5$ (the colouring is not indicated, but recall that~$u_i$ and~$v_i$ are coloured alike and otherwise colours are distinct).
The shaded areas represent a colourful partition of minimum size; there are 11 edges between the two parts.
The 10 highlighted edges are a smallest set whose deletion leaves a colourful graph.}
\end{figure}

\subsection{Known Results}\label{s-known}

Adamaszek and Popa~\cite{AP15} proved, among other results, that {\sc Colourful Partition} does not admit a polynomial-time approximation within a factor of~$n^{\frac{1}{14}-\epsilon}$, for any $\epsilon>0$ (assuming $\sf{P}\neq \NP$).
A coloured graph~$(G,c)$ is {\em $\ell$-coloured} if $1\leq c(u)\leq \ell$ for all $u\in V(G)$.
Bruckner et al.~\cite{BHKNTU12} proved the following two results for {\sc Colourful Components}.
The first result follows from observing that for $\ell=2$, the problem becomes a maximum matching problem in a bipartite graph after removing all edges between vertices coloured alike.
This observation can also be used for {\sc Colourful Partition}.

\begin{theorem}[\cite{BHKNTU12}]\label{t-2colours}
{\sc Colourful Partition} and {\sc Colourful Components} are polynomial-time solvable for $2$-coloured graphs.
\end{theorem}

\begin{theorem}[\cite{BHKNTU12}]\label{t-3colours-degree6}
{\sc Colourful Components} is \NP-complete for $3$-coloured graphs of maximum degree~$6$.
\end{theorem}
The situation for trees is different than for general graphs (see Example~\ref{e-intro}).
A tree~$T$ has a colourful partition of size~$k$ if and only if it can be modified into a colourful graph by at most~$k-\nobreak 1$ edge deletions.
Hence, the problems {\sc Colourful Partition} and {\sc Colourful Components} are equivalent for trees.
The following hardness and \FPT\ results are due to Bruckner et al.~\cite{BHKNTU12} and Dondi and Sikora~\cite{DS18}.
Note that trees of diameter at most~$3$ are stars and double stars (the graph obtained from two stars by adding an edge between their central vertices), for which both problems are readily seen to be polynomial-time solvable.
A {\em subdivided star} is the graph obtained by subdividing the edges of a star.

\begin{theorem}[\cite{BHKNTU12}]\label{t-trees-diameter4}
{\sc Colourful Partition} and {\sc Colourful Components} are polynomial-time solvable for coloured trees of diameter at most~$3$, but \NP-complete for coloured trees of diameter~$4$.
\end{theorem}

\begin{theorem}[\cite{DS18}]\label{t-paths}
{\sc Colourful Partition} and {\sc Colourful Components} are polynomial-time solvable for coloured paths (which have path-width at most~$1$), but \NP-complete for coloured subdivided stars (which are trees of path-width at most~$2$).
\end{theorem}

\begin{theorem}[\cite{BHKNTU12}]\label{t-trees-fpt-colours}
{\sc Colourful Partition} and {\sc Colourful Components} are \FPT\ for coloured trees, when parameterized by the number of colours.
\end{theorem}

\begin{theorem}[\cite{DS18}]\label{t-trees-fpt-components}
{\sc Colourful Partition} and {\sc Colourful Components} are \FPT\ for coloured trees, when parameterized by the number of colourful components (or equivalently, by the size of a colourful partition).
\end{theorem}
In fact, Dondi and Sikora~\cite{DS18} found a cubic kernel for coloured trees when parameterized by~$k$ and also gave an $O^*(1.554^k)$-time exact algorithm for coloured trees.
In addition to Theorem~\ref{t-trees-fpt-colours}, Bruckner et al.~\cite{BHKNTU12} showed that {\sc Colourful Components} is \FPT\ for general coloured graphs when parameterized by the number of colours~$\ell$ and the number of edge deletions~$p$.

Misra~\cite{Mi18} considered the size of a vertex cover as the parameter.

\begin{theorem}[\cite{Mi18}]\label{t-vertexcover_c}
{\sc Colourful Components} is \FPT\ when parameterized by vertex cover number.
\end{theorem}

\subsection{Our Aims and Results}\label{s-our}
Our main focus is the {\sc Colourful Partition} problem, and our research aims are:
\begin{enumerate}
\item to tighten known \NP-hardness results and generalize tractability results for {\sc Colourful Partition} (Section~\ref{s-hard});
\item to prove new results analogous to those known for {\sc Colourful Components} (Section~\ref{s-hard}) and to show that the two problems are different (Sections~\ref{s-hard} and~\ref{s-fpt}); and
\item
to use the new results to determine suitable parameters for obtaining \FPT\ results (Section~\ref{s-fpt}).
\end{enumerate}
Along the way we also prove some new results for {\sc Colourful Components}.

First, we show an analogue of Theorem~\ref{t-3colours-degree6} by proving that {\sc Colourful Partition} is \NP-complete even for $3$-coloured $2$-connected planar graphs of maximum degree~$3$.
This result is best possible, as {\sc Colourful Partition} is polynomial-time solvable for $2$-coloured graphs (Theorem~\ref{t-2colours}) and for graphs of maximum degree~$2$ (trivial).
We show that it also gives us a family of instances on which {\sc Colourful Components} and {\sc Colourful Partition} have different complexities.

Second, we focus on coloured trees.
Due to Theorem~\ref{t-trees-fpt-colours}, {\sc Colourful Partition} is polynomial-time solvable for $\ell$-coloured trees for any fixed~$\ell$.
Hence the coloured trees in the hardness proofs of Theorems~\ref{t-trees-diameter4} and~\ref{t-paths} use an arbitrarily large number of colours.
They also have arbitrarily large degree.
Moreover, they have arbitrarily large degree.
We define the {\em colour-multiplicity} of a coloured graph~$(G,c)$ as the maximum number of vertices in~$G$ with the same colour.
We prove that {\sc Colourful Partition} is \NP-complete even for coloured trees of maximum degree~$6$ and colour-multiplicity~$2$.
As both problems are equivalent on trees, we obtain the same result for {\sc Colourful Components} (note that the graphs in the proof of~Theorem~\ref{t-3colours-degree6}~are~not~trees).

Third, we fix the number of colourful components.
Note that {\sc $k$-Colourful Partition} is polynomial-time solvable for coloured trees (due to Theorem~\ref{t-trees-fpt-components}) and
for $\ell$-coloured graphs for any fixed~$\ell$ (such graphs have at most~$k\ell$ vertices).
We prove that {\sc $2$-Colourful Partition} is \NP-complete for split graphs and for coloured planar bipartite graphs of maximum degree~$3$ and path-width~$3$, but polynomial-time solvable for coloured graphs of treewidth at most~$2$.
The latter two results also complement Theorem~\ref{t-paths}, which implies \NP-completeness of {\sc Colourful Partition} for path-width~$2$.

In Section~\ref{s-fpt} we show that {\sc Colourful Partition} and {\sc Colourful Components} are \FPT\ when parameterized by the treewidth and the number of colours, which generalizes Theorem~\ref{t-trees-fpt-colours}.
Our choice for this combination of parameters was guided by our results from Section~\ref{s-hard}.
As we will argue at the start of Section~\ref{s-fpt}, for the six natural parameters: number of colourful components; maximum degree; number of colours; colour-multiplicity; path-width; treewidth; and all their possible combinations, we now obtained either para-NP-completeness or an \FPT\ algorithm for {\sc Colourful Partition}.
This motivates the search for other parameters for this problem.
As the vertex cover number of the subdivided stars in the proof of Theorem~\ref{t-paths} can be arbitrarily large, it is natural to consider this parameter.
We provide an analogue to Theorem~\ref{t-vertexcover_c} by proving that {\sc Colourful Partition} is \FPT\ when parameterized by vertex cover number.
A vertex~$u$ of a coloured graph~$(G,c)$ is {\em uniquely coloured} if its colour~$c(u)$ is not used on any other vertex of~$G$.
It is easy to show \NP-hardness for {\sc Colourful Partition} and {\sc Colourful Components} for instances with no uniquely coloured vertices (see also Theorem~\ref{t-trees2}).
We also prove that {\sc Colourful Components} is para-\NP-hard when parameterized by the number of non-uniquely coloured vertices, whereas {\sc Colourful Partition} is \FPT\ in this case.
Thus there are families of instances on which the two problems have different parameterized complexities.

\section{Preliminaries}\label{s-pre}

All graphs in this paper 
are simple, with no loops or multiple edges.
Let $G=(V,E)$ be a graph.
A subset $U\subseteq V$ is {\em connected} if it induces a connected subgraph of~$G$.
For a vertex $u\in V$, $N(u)=\{v\; |\; uv\in E\}$ is the {\em neighbourhood} of~$u$, and $\deg(u)=|N(u)|$ is the {\em degree} of~$u$.
A tree is {\em binary} if every vertex in it has degree at most~$3$.
A graph is {\em cubic} if every vertex has degree exactly~$3$.
A connected graph on at least three vertices is {\em $2$-connected} if it 
does not contain a 
vertex whose removal disconnects the graph.
We say that a
graph~$G=(V,E)$ is {\em split} if~$V$ can be partitioned into two (possibly empty) sets~$K$ and~$I$, where~$K$ is a clique and~$I$ is an independent set.
A mapping $c:E\to \{1,2,3\}$ is a {\em proper $3$-edge colouring} of~$G$ if $c(e)\neq c(f)$ for any two distinct edges~$e$ and~$f$ 
that share 
a common end-vertex.
A set $S\subseteq V$ is a {\em vertex cover}
 of a graph $G=(V,E)$
if $G-S$ is an independent set.
The {\sc Vertex Cover} problem asks if a 
given 
graph has a vertex cover of size at most~$s$ for a given integer~$s$.
The {\em vertex cover number}~$\vc(G)$ of a graph~$G$ is the minimum size of a vertex cover in~$G$.
We 
will use 
the following lemma.

\begin{lemma}[\cite{BFR13}]\label{l-vertexcover}
{\sc Vertex Cover} is \NP-complete for $2$-connected cubic planar graphs with a proper $3$-edge colouring given as input.
\end{lemma}
A {\em tree decomposition} of a graph~$G$ is a pair $(T,{\cal X})$ where~$T$ is a tree and ${\cal X}=\{X_{i} \mid i\in V(T)\}$ is a collection of subsets of~$V(G)$, called {\em bags}, such that

\begin{enumerate}
\item $\bigcup_{i \in V(T)} X_{i} = V(G)$;
\item for every edge $xy \in E(G)$, there is an $i \in V(T)$ such that $x,y\in X_i$; and
\item for every $x\in V(G)$, the set $\{ i\in V(T) \mid x \in X_{i} \}$ induces a connected subtree of~$T$.
\end{enumerate}

\noindent
The {\em width} of~$(T,{\cal X})$ is $\max\{|X_{i}| - 1 \;|\; i \in V_T\}$, and the {\em treewidth} of~$G$ is the minimum width over all tree decompositions of~$G$.
If~$T$ is a path, then~$(T,{\cal X})$ is a {\em path decomposition} of~$G$.
The {\em path-width} of~$G$ is the minimum width over all path decompositions of~$G$.

A tree decomposition $(T,{\cal X})$ of a graph~$G$ is \emph{nice} if~$T$ is a rooted binary tree such that the nodes of~$T$ are of four types:

\begin{enumerate}
\item a \emph{leaf node}~$i$ is a leaf of~$T$ with $X_i=\emptyset$;
\item an \emph{introduce node}~$i$ has one child~$i'$ with $X_i=X_{i'}\cup\{v\}$ for some vertex $v\in V(G)$;
\item a \emph{forget node}~$i$ has one child~$i'$ with $X_i=X_{i'}\setminus\{v\}$ for some vertex $v\in V(G)$; and
\item a \emph{join node}~$i$ has two children~$i'$ and~$i''$ with $X_i=X_{i'}=X_{i''}$,
\end{enumerate}

\noindent
and, moreover, the root~$r$ is a forget node with $X_r=\emptyset$.
Kloks~\cite{Kl94} proved that every tree decomposition of a graph can be converted in linear time to a nice tree decomposition of the same width such that the size of the obtained tree is linear in the size of the original tree.

A \emph{literal} is a (propositional) variable~$x$ ({\em positive} literal) or a negated variable~$\ol{x}$ ({\em negative} literal).
A set~$S$ of literals is \emph{tautological} if $S\cap \ol{S}\neq \emptyset$, where $\ol{S}=\{\ol{x} \;|\; x\in S\}$.
A \emph{clause} is a finite non-tautological set of literals.
A \emph{(CNF) formula} is a finite set of clauses.
For $k\geq 1$, a \emph{$k$-formula} is a formula in which every clause contains exactly~$k$ pairwise distinct variables.
A $k$-formula is {\em positive} if every clause in it contains only positive literals.
A \emph{truth assignment} for a set of variables~$X$ is a mapping $\tau:X\rightarrow \{\True,\False\}$.
The {\sc $k$-Satisfiability} problem asks whether there exists a truth assignment for a given $k$-formula~$F$ such that every clause of~$F$ contains at least one true literal.
The {\sc Not-All-Equal Positive $3$-Satisfiability} problem asks whether there exists a truth assignment for a given positive $3$-formula~$F$ such that every clause of~$F$ contains at least one true literal and at least one false literal.
For both problems we say that such a desired truth assignment {\em satisfies}~$F$.
Both {\sc $3$-Satisfiability}~\cite{Ka72} and {\sc Not-All-Equal Positive $3$-Satisfiability}~\cite{Sc78} are \NP-complete.
In contrast, it is well known and readily seen that {\sc $2$-Satisfiability} can be solved in polynomial time~\cite{Kr67}.

Let~$X$ be a set.
Then ${\cal P}(X)=\{Y \;|\; Y \subseteq X\}$ is the {\em power set} of~$X$, that is, the set of subsets of~$X$.
A {\em partition}~$P$ of~$X$ is a subset of~${\cal P}(X)$ such that $\emptyset \notin P$ and every element of~$X$ is in exactly one set in~$P$.
A {\em block} is an element of a partition~$P$ and the {\em size} of~$P$ is~$|P|$.
Given two partitions~$P$ and~$Q$, we say that~$P$ is {\em coarsening} of~$Q$ if for every $Y \in Q$ there is a $Y' \in P$ such that $Y \subseteq Y'$.
Let $\varnothing:\emptyset \to \emptyset$ denote the trivial function.

\section{Classical Complexity}\label{s-hard}
We will prove four hardness results and one polynomial-time result on {\sc Colourful Partition}.
Note that {\sc Colourful Partition} belongs to \NP.
We start by proving the following result.

\begin{theorem}\label{t-33}
{\sc Colourful Partition} is \NP-complete for $3$-coloured $2$-connected planar graphs of maximum degree~$3$.
\end{theorem}

\begin{proof}
We use a reduction from {\sc Vertex Cover}.
By Lemma~\ref{l-vertexcover} we may assume that we are given a $2$-connected cubic planar graph~$G$ with a proper $3$-edge colouring~$c$.
From~$G$ and~$c$ we construct a coloured graph~$(G',c')$ as follows.

\vspace{-10pt}

\begin{itemize}
\item For each $e\in E(G)$ with $c(e)=i$, create a vertex~$v_e$ with colour $c'(v_e)=i$ in~$G'$ (a {\em red} vertex).
\item For every vertex $v\in V(G)$ with incident edges $e_1$, $e_2$, $e_3$ (with colours $1$, $2$, $3$, respectively), create a copy of the $3$-coloured $9$-vertex gadget shown in \figurename~\ref{fig:VC-reduction} (a {\em blue} set), and connect it to the red vertices $v_{e_1}$, $v_{e_2}$ and~$v_{e_3}$, as shown in the same figure.
\end{itemize}

\begin{figure}
\centering
\begin{tikzpicture}
\vertexGadgetThreeComponents
\end{tikzpicture}

\caption{\label{fig:VC-reduction}
The blue $9$-vertex gadget for a vertex~$v$ incident to edges $e_1,e_2,e_3$ in~$G$, connected to the three red vertices $v_{e_1},v_{e_2},v_{e_3}$ in~$G'$ in the proof of Theorem~\ref{t-33}.
A crucial property is that the three red vertices have distinct colours, as~$c$ is a $3$-edge colouring of~$G$, and these three colours fix the colours of the blue vertex-gadget.}
\end{figure}
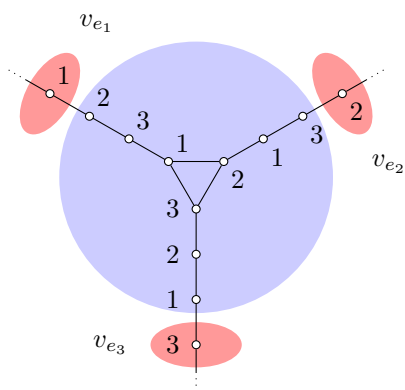

Note that~$G'$ is $3$-coloured, $2$-connected, planar and has maximum degree~$3$.
We claim that~$G$ has a vertex cover of size at most~$s$ if and only if~$G'$ has a colourful partition of size at most~$3n+\nobreak s$.

First suppose that~$G$ has a vertex cover~$S$ of size at most~$s$.
We create the colourful partition of~$G'$ as described in \figurename~\ref{fig:VC-reduction-selection}.
More precisely, for every vertex $v\in S$, we pick four colourful components covering the whole vertex gadget plus all neighbouring red vertices that have not yet been picked.
Since~$S$ is a vertex cover of~$G$, all red vertices of~$G'$ are now part of some colourful component.
For every vertex $v\in V(G)\setminus S$, we pick three colourful components covering the whole vertex gadget.
This yields a colourful partition of~$G'$ of size at most $4|S|+3(n-|S|)=3n+|S| \leq 3n+s$.

Now suppose that~$G$ has a colourful partition $(V_1,\ldots,V_k)$ for some $k\leq 3n+s$.
Observe that no~$V_i$ contains vertices from two distinct vertex gadgets; otherwise~$V_i$ would contain a red vertex and its two neighbours, which have the same colour by construction.
Hence, every~$V_i$ can be mapped to some vertex~$v_j$ of~$G$ (where we map every colourful component that consists of only a single red vertex to an arbitrary adjacent vertex gadget).
Observe that $|V_i|\leq 3$ for all $i \in \{1,\ldots,k\}$.
As every vertex gadget has size~$9$, at least three colourful components are mapped to it.
Moreover, there cannot be two neighbouring gadgets with only three components, as the two of them plus the red vertex connecting them have 19 vertices in total.
So the vertices of~$G$ corresponding to vertex gadgets with at least four components form a vertex cover~$S$ in~$G$.
Since there are at most $k-3n\leq s$ such vertex gadgets, $S$ has size at most~$s$.
\end{proof}

The coloured graph~$(G',c')$ constructed in the proof of Theorem~\ref{t-33} can be modified into a colourful graph by omitting exactly one edge adjacent to each red vertex and exactly three edges inside each blue component.
It can be readily checked that this is the minimum number of edges required.
Hence {\sc Colourful Components} is polynomial-time solvable on these coloured graphs~$(G',c')$, whereas {\sc Colourful Partition} is \NP-complete.
Thus we have:

\begin{corollary}\label{c-diff}
There exists a family of instances on which {\sc Colourful Components} and {\sc Colourful Partition} have different complexities.
\end{corollary}

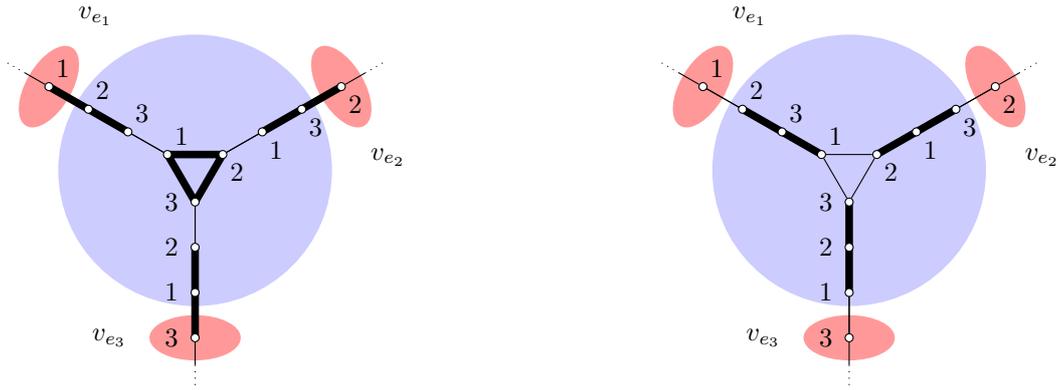
\begin{figure}
\begin{tikzpicture}
\begin{scope}[scale=.6, label distance=1mm]
\foreach \d/\l/\a/\ca/\cb/\cc in {150/60/Za/1/3/2, 270/180/Zb/3/2/1,  30/300/Zc/2/1/3 } {
 \path [] (\d:.7) node[minimum size=0pt] (\a1) {}
        -- ++ (\d:1)node[minimum size=0pt] (\a2) {}
        -- ++ (\d:1)node[minimum size=0pt] (\a3) {}
        -- ++ (\d:1)node[minimum size=0pt] (\a4) {};}

\path[line width=1mm] (Za2)--(Za3)--(Za4);
\path[line width=1mm] (Zb2)--(Zb3)--(Zb4);
\path[line width=1mm] (Zc2)--(Zc3)--(Zc4);
\path[line width=1mm] (Za1)--(Zb1)--(Zc1)--(Za1);
\end{scope}
\vertexGadgetThreeComponents

\end{tikzpicture}\hfill\begin{tikzpicture}
\begin{scope}[scale=.6, label distance=1mm]
\foreach \d/\l/\a/\ca/\cb/\cc in {150/60/Za/1/3/2, 270/180/Zb/3/2/1,  30/300/Zc/2/1/3 } {
 \path [] (\d:.7) node[minimum size=0pt] (\a1) {}
        -- ++ (\d:1)node[minimum size=0pt] (\a2) {}
        -- ++ (\d:1)node[minimum size=0pt] (\a3) {}
        -- ++ (\d:1)node[minimum size=0pt] (\a4) {};}

\path[line width=1mm] (Za1)--(Za2)--(Za3);
\path[line width=1mm] (Zb1)--(Zb2)--(Zb3);
\path[line width=1mm] (Zc1)--(Zc2)--(Zc3);
\end{scope}
\vertexGadgetThreeComponents

\end{tikzpicture}
\caption{\label{fig:VC-reduction-selection}The colourful components (in bold) corresponding to a vertex in the vertex cover~$S$ of~$G$ (left figure) and to a vertex not in~$S$ (right figure) in the proof of Theorem~\ref{t-33}.}
\end{figure}

For our next result we need to introduce the following problem.

\problemdef{Multicut}{A graph~$G$, a positive integer~$r$ and a set~$S$ of pairs of distinct vertices of~$G$.}{Is there a set of at most~$r$ edges whose removal from~$G$ results in a graph in which~$s$ and~$t$ belong to different connected components for every pair $(s,t) \in S$?}

\noindent If $(s,t)$ is a pair in~$S$, we say that~$s$ and~$t$ are \emph{mates}. Moreover, every vertex that belongs to a pair in~$S$ is a \emph{terminal}. The {\sc Multicut} problem is known to be \NP-complete even for binary trees~\cite{CFR03}.

We are now ready to prove the following theorem. In particular, the last condition in this theorem shows that the number of uniquely coloured vertices is not a useful parameter.

\begin{theorem} \label{t-trees2}
{\sc Colourful Partition} and {\sc Colourful Components} are \NP-complete for coloured trees with maximum degree at most~$6$, colour-multiplicity~$2$, and no uniquely coloured vertices.
\end{theorem}

\begin{proof}
As the two problems are equivalent on trees, it suffices to prove the result for {\sc Colourful Partition}.
We first show the result for the case where we allow uniquely coloured vertices.
Afterwards we show how to modify our construction to get rid of such vertices.

As mentioned, we reduce from {\sc Multicut} on binary trees. 
Let $(G,S,r)$ be an instance of {\sc Multicut}, where~$G$ is a binary tree, $S$ is a set of terminal pairs and~$r$ is positive integer.
We construct a graph~$G'$ that has~$G$ as a subgraph and some additional vertices and edges that we now describe.
For a vertex $v\in V(G)$, let $e_v^1,\ldots,e_v^{\deg(v)}$ be the edges incident with~$v$ and let $A_v^1,\ldots,A_v^{\deg(v)}$ be a collection of sets whose contents we now define.
For each mate~$u$ of~$v$ in the pairs of~$S$, create a vertex~$v_u$ is and place it in the set~$A_v^i$ such that the unique path in~$G$ from~$v$ to~$u$ includes the edge~$e_v^i$.
For each~$A_v^i$, add edges so that the vertices in~$A_v^i$ induce a path, and add an edge from one end-vertex of this path to~$v$.
Let~$P_v^i$ be the path that includes~$v$ and the vertices of~$A_v^i$.
This completes the construction of~$G'$.
As~$G$ is binary, $G'$ is a tree of maximum degree at most~$6$.
We now construct a colouring~$c$ on the vertices of~$G'$.
For each pair~$(u,v)$ in~$S$, we let $c(v_u)=c(u_v)$ be a colour that is not used on any other vertex of~$G'$.
All other vertices of~$G'$ are assigned a unique colour.
Hence~$(G',c)$ has colour-multiplicity~$2$.
See \figurename~\ref{fig:multicut} for an illustration.

\begin{figure}
\begin{tikzpicture}
 \path [thick]  (0,0) node (g)[label=left:$g$]{}
 	   --  ++(0.75,1.5) node (d)  [label=left:$d$] {}
 	   -- ++(0.75,1.5) node (b) [label=left:$b$] {}
 	   -- ++(0.75,1.5) node (a) [label=above:$a$] {}
 	   -- ++(0.75,-1.5) node (c) [label=right:$c$] {}
 	   -- ++(0.75,-1.5) node (f) [label=right:$f$] {}
 	   -- ++(0.75,-1.5) node (j) [label=right:$j$] {};

\path [thick] (b)
 	   -- ++(0.75,-1.5) node (e) [label=right:$e$] {}
 	   -- ++(-0.75,-1.5) node (h) [label=left:$h$] {};

\path [thick] (e)
 	   -- ++(0.75,-1.5) node (i) [label=right:$i$] {};

\node[draw=none,rectangle] at (0.25,2.5) {$G$};

\node[draw=none,rectangle,align=right] at (2.2,-1) {$S=[(a,c), (a,h), (a,j), (b,c), (b,d),$ \\ $(b,h), (b,i), (b,j), (c,d), (c,j)]$};

\begin{scope}[xshift=7cm]
 \path [thick]  (0,0) node (g)[label=left:$g$]{}
 	   --  ++(0.75,1.5) node (d)  [label=right:$d$] {}
 	   -- ++(0.75,1.5) node (b) [label=right:$b$] {}
 	   -- ++(0.75,1.5) node (a) [label=above:$a$] {}
 	   -- ++(0.75,-1.5) node (c) [label=left:$c$] {}
 	   -- ++(0.75,-1.5) node (f) [label=right:$f$] {}
 	   -- ++(0.75,-1.5) node (j) [label=left:$j$] {};

\path [thick] (b)
 	   -- ++(0.75,-1.5) node (e) [label=right:$e$] {}
 	   -- ++(-0.75,-1.5) node (h) [label=right:$h$] {};

\path [thick] (e)
 	   -- ++(0.75,-1.5) node (i) [label=left:$i$] {};

\path [thick] (a)
 	   -- ++(0.5,0.25) node [label=above:$a_c$] {}
 	   -- ++(0.5,0) node [label=above:$a_j$] {};

\path [thick] (a)
 	   -- ++(-0.5,0.25) node [label=above:$a_h$] {};

\path [thick] (b)
 	   -- ++(-0.75,0.5) node [label=above:$b_c$] {}
 	   -- ++(-0.5,0) node [label=above:$b_j$] {};

\path [thick] (b)
 	   -- ++(-0.75,0) node [label=left:$b_d$] {};

\path [thick] (b)
 	   -- ++(-0.75,-0.5) node [label=below:$b_i$] {}
 	   -- ++(-0.5,0) node [label=below:$b_h$] {};

\path [thick] (d)
 	   -- ++(-0.5,0) node [label=below:$d_b$] {}
 	   -- ++(-0.5,0) node [label=below:$d_c$] {};

\path [thick] (h)
 	   -- ++(-0.5,-0.25) node [label=below:$h_a$] {}
 	   -- ++(-0.5,0) node [label=below:$h_b$] {};

\path [thick] (i)
 	   -- ++(0.5,-0.25) node [label=below:$i_b$] {};

\path [thick] (c)
 	   -- ++(0.75,0.5) node [label=above:$c_a$] {}
 	   -- ++(0.5,0) node [label=above:$c_b$] {}
 	   -- ++(0.5,0) node [label=above:$c_d$] {};

\path [thick] (c)
 	   -- ++(0.75,-0.5) node [label=below:$c_j$] {};

\path [thick] (j)
 	   -- ++(0.5,-0.25) node [label=below:$j_a$] {}
 	   -- ++(0.5,0) node [label=below:$j_b$] {}
 	   -- ++(0.5,0) node [label=below:$j_c$] {};

\node[draw=none,rectangle] at (5.5,2.5) {$G'$};

\end{scope}

\end{tikzpicture}
\caption{\label{fig:multicut}A graph~$G$ and a set of terminals~$S$ of an instance of {\sc Multicut} and the instance~$G'$ of {\sc Colourful Partition} obtained using the reduction in the proof of Theorem~\ref{t-trees2}.
The colouring of~$G'$ is not indicated and the value of~$r$, which does not affect the construction, is not stated.}
\end{figure}
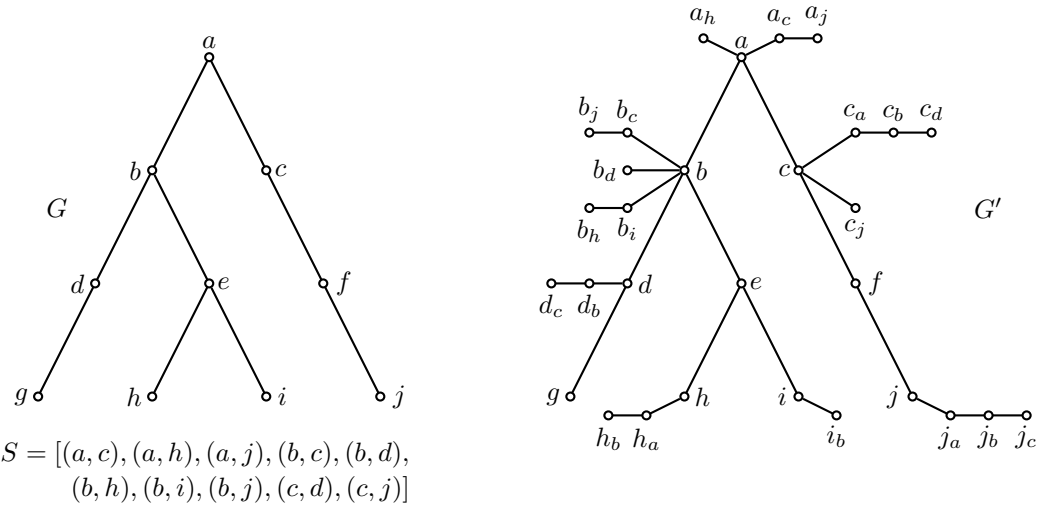

We claim that it is possible to separate all mates from each other by removing at most~$r$ edges from~$G$ if and only if~$(G',c)$ has a colourful partition of size at most~$r+1$.

First suppose that it is possible to separate all mates of~$S$ from each other by removing at most~$r$ edges from~$G$.
We remove the same~$r$ edges from~$G'$ (note that all edges in~$G$ are also in~$G'$).
The resulting graph consists of at most~$r+\nobreak 1$ components.
Let $V_1,\ldots,V_k$, $k \leq r+1$, be the vertex sets of these components.
Note that vertices~$u_v$ and~$v_u$ belong to different sets~$V_i$ and~$V_j$, as~$u$ and~$v$, being mates, must belong to different components, while~$v_u$ belongs to the same component as~$v$, and~$u_v$ belongs to the same component as~$u$.
Hence, $(V_1,\ldots,V_k)$ is a colourful partition of~$G'$ that has size at most~$r+\nobreak 1$.

Now suppose that~$G'$ has a colourful partition $(V_1,\ldots,V_k)$, $k \leq r+1$.
As~$G'$ is a tree, we can obtain this partition by removing a set~$F$ of at most~$r$ edges from~$G'$.
If~$F$ contains edges not in~$E(G)$, then we show that we can modify~$F$ to obtain a set~$F^*$ of at most~$r$ edges, all of which belong to~$E(G)$, such that removing the edges in~$F^*$ from~$G'$ leaves a graph whose connected components form a colourful partition.
To show this, suppose~$F$ contains an edge that belongs to some~$P_w^i$.
In~$F$ we replace \emph{all} the edges of~$P_w^i$ with the edge~$e_w^i$ to obtain a set~$F'$ with $|F'| \leq |F|$.
We claim that if a pair of vertices~$v_u$ and~$u_v$ are separated by~$F$, then they are separated by~$F'$ and so removing the edges in~$F'$ from~$G'$ also leaves a colourful partition.
This immediately holds if neither~$u_v$ nor~$v_u$ are in~$P_w^i$.
If one of $u_v$,~$v_u$ does belong to~$P_w^i,$ then $w \in \{u,v\}$ and the unique path from~$u$ to~$v$ includes~$e_w^i$, which therefore separates~$u_v$ and~$v_u$.
Thus deleting the vertices of~$F'$ also gives a colourful partition of size at most~$r+\nobreak 1$.
Furthermore, $F'$ has a larger intersection with~$E(G)$ than~$F$ does.
Hence, by repetition of this modification, we obtain~$F^*$.
Finally we demonstrate that~$F^*$ separates all mates of~$S$.
For each pair $(u,v) \in S$, we know that~$v$ and~$v_u$ belong to the same component of $G-F^*$, and that~$u$ and~$u_v$ belong to the same component of $G-F^*$ (as~$F^*$ contains only edges of~$G$).
Thus as~$F^*$ separates the like-coloured pair~$u_v$ and~$v_u$, it also separates the mates~$u$ and~$v$.

\medskip
\noindent
We now prove the same result with the addition that there are no uniquely coloured vertices.
Let~$(G',c)$ be an instance of {\sc Colourful Partition} constructed above and note that each colour is used at most twice in~$c$.
We take two copies~$G_1$ and~$G_2$ of the graph~$G'$ and for $v \in V(G)$ let~$v_1$ and~$v_2$ be the corresponding vertices in~$G_1$ and~$G_2$.
We define a new colouring~$c'$ on the disjoint union $G_1+\nobreak G_2$ as follows:
\begin{itemize}
\item If~$c$ assigns a colour to a unique vertex~$v$ of~$G$, then let $c'(v_1)=c'(v_2)=c(v)$.
\item If~$c$ assigns the same colour to two distinct vertices $u,v$ in~$G$, then let $c'(u_1)=c'(v_1)=c(u)$ and let~$c'$ assign a new colour to~$u_2$ and~$v_2$ that is not used anywhere else.
\end{itemize}
It is easy to verify that $(G_1+\nobreak G_2,c')$ has a colourful partition of size at most~$2(r+\nobreak 1)$ if and only if~$(G',c)$ has a colourful partition of size at most $r+\nobreak 1$.

Finally, let~$x$ be a vertex in~$G'$ that has degree~$1$ (such a vertex must exist if~$G'$ is a tree with at least one edge).
We let~$G''$ be the graph obtained from $G_1+\nobreak G_2$ by adding vertices~$y_1$ and~$y_2$ and edges $x_1y_1$, $y_1y_2$ and~$y_2x_2$.
Let~$c''$ be the colouring that colours~$y_1$ and~$y_2$ with a new colour that is not used anywhere else and colours all other vertices of~$G''$ in the same way as~$c'$ does.
Note that~$G''$ is a tree of maximum degree at most~$6$ and~$c''$ colours the vertices of~$G''$ such that every colour is used exactly twice.
Furthermore, in every colourful partition of $(G'',c'')$, the vertices~$y_1$ and~$y_2$ must be in different parts of the partition.
Since the colours on~$y_1$ and~$y_2$ are not used anywhere else, in any colourful partition of minimum size~$y_1$ must be in the same colourful component as~$x_1$ and~$y_2$ must be in the same colourful component at~$x_2$.
Thus $(G'',c'')$ has a colourful partition of size at most~$2(r+\nobreak 1)$ if and only if~$(G',c)$ has a colourful partition of size at most $r+\nobreak 1$.
\end{proof}

We now present our third \NP-hardness result.

\begin{theorem}\label{t-split}
{\sc $2$-Colourful Partition} is \NP-complete for coloured split graphs.
\end{theorem}

\begin{proof}
We use a reduction from the \NP-complete problem {\sc $3$-Satisfiability}~\cite{Ka72}.
Let $(X,{\cal C})$ be a $3$-formula, where $X= \{x_1,x_2,\ldots,x_n\}$ and ${\cal C} = \{C_1, C_2,\ldots, C_m\}$.
From $(X,{\cal C})$ we construct a coloured graph~$(G,c)$ as follows.

\begin{itemize}
\item For every~$x_i$ create two vertices~$x_i$,~$\ol{x_i}$; we say that these vertices are of {\em $x$-type}.
\item Add an edge between every pair of $x$-type vertices.
\item For every~$x_i$ create two vertex~$y_i$ and~$y_i'$ and add the edges $x_iy_i$, $x_iy_i'$, $\ol{x_i}y_i$, $\ol{x_i}y_i'$ (the vertices~$x_i$ and~$\ol{x_i}$ will be the only two neighbours of~$y_i$ and~$y_i'$).
\item Create a vertex~$z$ and add an edge between~$z$ and every $x$-type vertex.
\item For every clause~$C_j$ create a vertex~$C_j$ and a vertex~$C_j'$.
\item Add an edge between~$z$ and every~$C_j'$ (every~$C_j'$ vertex will have~$z$ as its only neighbour).
\item If a clause~$C_j$ consists of literals $x_g$, $x_h$, $x_i$, then add the edges $C_jx_g$, $C_jx_h$ and~$C_jx_i$.
\item Assign a distinct colour to every $x$-type vertex, every~$y_i$ vertex, every clause vertex~$C_j$ and~$z$.
\item Set $c(y_i')=c(y_i)$ for $i\in\{1,\ldots,n\}$ and $c(C_j')=c(C_j)$ for $j\in\{1,\ldots,m\}$.
\end{itemize}
See \figurename~\ref{fig:split} for an illustration. We note that~$G$ is a split graph, as the set of all $x$-type vertices and~$z$ is a clique, whereas the other vertices form an independent set.

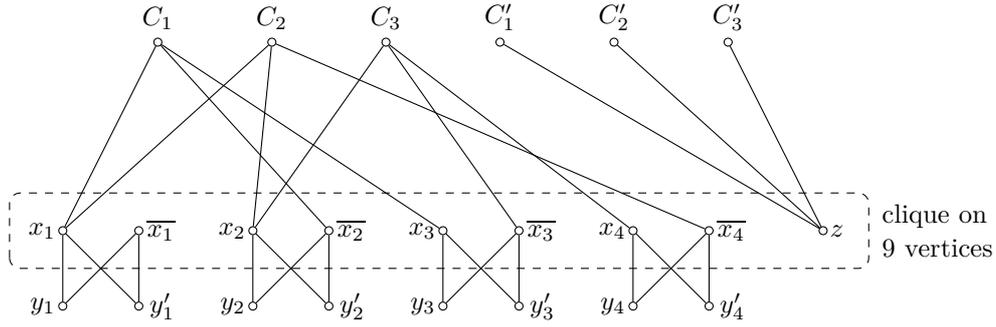
\begin{figure}
\begin{center}
\begin{tikzpicture}

\node[label=left:$x_1$] (x1t) at (0,0) {};
\node[label=right:$\ol{x_1}$] (x1f) at (1,0) {};
\node[label=left:$x_2$] (x2t) at (2.5,0) {};
\node[label=right:$\ol{x_2}$] (x2f) at (3.5,0) {};
\node[label=left:$x_3$] (x3t) at (5,0) {};
\node[label=right:$\ol{x_3}$] (x3f) at (6,0) {};
\node[label=left:$x_4$] (x4t) at (7.5,0) {};
\node[label=right:$\ol{x_4}$] (x4f) at (8.5,0) {};
\node[label=right:$z$] (z) at (10,0) {};

\draw[rounded corners, dashed] (-0.7,-0.5) rectangle (10.6,0.5);

\node[draw=none,rectangle,align=left] at (11.5,0.0) {clique on\\ 9 vertices};

\node[label=left:$y_1$] (y1t) at (0,-1) {};
\node[label=right:$y_1'$] (y1f) at (1,-1) {};
\node[label=left:$y_2$] (y2t) at (2.5,-1) {};
\node[label=right:$y_2'$] (y2f) at (3.5,-1) {};
\node[label=left:$y_3$] (y3t) at (5,-1) {};
\node[label=right:$y_3'$] (y3f) at (6,-1) {};
\node[label=left:$y_4$] (y4t) at (7.5,-1) {};
\node[label=right:$y_4'$] (y4f) at (8.5,-1) {};

\path (x1t) -- (y1t) -- (x1f) -- (y1f) -- (x1t); 
\path (x2t) -- (y2t) -- (x2f) -- (y2f) -- (x2t); 
\path (x3t) -- (y3t) -- (x3f) -- (y3f) -- (x3t); 
\path (x4t) -- (y4t) -- (x4f) -- (y4f) -- (x4t); 

\node[label=above:$C_1$] (c1) at (1.25,2.5) {};
\node[label=above:$C_2$] (c2) at (2.75,2.5) {};
\node[label=above:$C_3$] (c3) at (4.25,2.5) {};
\node[label=above:$C_1'$] (c1a) at (5.75,2.5) {};
\node[label=above:$C_2'$] (c2a) at (7.25,2.5) {};
\node[label=above:$C_3'$] (c3a) at (8.75,2.5) {};

\path (c1a) -- (z) -- (c2a);
\path (c3a) -- (z);

\path (c1) -- (x1t);
\path (c1) -- (x2f);
\path (c1) -- (x3t);

\path (c2) -- (x1t);
\path (c2) -- (x2t);
\path (c2) -- (x4f);

\path (c3) -- (x2t);
\path (c3) -- (x3f);
\path (c3) -- (x4t);

\end{tikzpicture}
\end{center}
\caption{\label{fig:split}The split graph~$G$ constructed using the reduction of the proof of Theorem~\ref{t-split} from the {\sc $3$-Satisfiability} instance $C_1=(x_1\vee \ol{x_2} \vee x_3), C_2=(x_1\vee {x_2} \vee \ol{x_4}), C_3= (x_2\vee \ol{x_3} \vee x_4)$.
The colouring is not indicated.}
\end{figure}

\medskip
\noindent
We claim that~${\cal C}$ has a satisfying truth assignment if and only if~$(G,c)$ has a colourful partition of size~$2$.

First suppose that~${\cal C}$ has a satisfying truth assignment~$\tau$.
We put an $x$-type vertex in~$V_1$ if and only if~$\tau$ assigns~$\True$ to it.
We also put every~$C_j$ and every~$y_i$ in~$V_1$.
As~$\tau$ is a satisfying truth assignment, $V_1$ is connected.
We let~$V_2$ consist of all other vertices and note that by construction~$V_2$ is also connected.
Moreover, both~$V_1$ and~$V_2$ are colourful.

Now suppose that~$(G,c)$ has a colourful partition~$(V_1,V_2)$.
We assume without loss of generality that $z\in V_2$.
Suppose there exists a vertex~$C_j'$ that belongs to~$V_1$.
Then, as~$C_j'$ has~$z$ as its only neighbour, $V_1$ only contains~$C_j'$, which means that~$V_2$ contains~$y_1$ and~$y_1'$, which have the same colour, a contradiction.
Hence, every~$C_j'$ must belong to~$V_2$.
As for every $j\in\{1,\ldots,m\}$ the two vertices~$C_j$ and~$C_j'$ have the same colour, this implies that every~$C_j$ belongs to~$V_1$.
Suppose that there exist two vertices~$x_i$ and~$\ol{x_i}$ that both belong to~$V_1$.
As~$y_i$ and~$y_i'$ have the same colour, one of them must belong to~$V_2$.
However, this is not possible, as~$V_2$ is connected and contains~$z$ (and every~$C_j'$), but~$x_i$ and~$\ol{x_i}$ are the only neighbours of~$y_i$ and~$y_i'$.
Suppose that there exist two vertices~$x_i$ and~$\ol{x_i}$ that both belong to~$V_2$.
Then, as at least one of $y_i, y_i'$ belongs to~$V_1$, we find that~$V_1$, which contains every~$C_j$, is not connected.
Hence this is also not possible.
We conclude that~$x_i$ and~$\ol{x_i}$ belong to different sets $V_1$,~$V_2$.
We can define a truth assignment~$\tau$ with $\tau(x)=\True$ if~$x$ belongs to~$V_1$ and $\tau(x)=\False$ if~$x$ belongs to~$V_2$.
As~$V_1$ is connected, each clause~$C_j$ is adjacent to at least one $x$-type vertex in~$V_1$.
So every clause~$C_j$ contains at least one literal to which~$\tau$ assigns \True.
\end{proof}

We now prove our final \NP-hardness result.

\begin{theorem}\label{t-pathwidth}
{\sc $2$-Colourful Partition} is \NP-complete for coloured planar bipartite graphs of maximum degree~$3$ and path-width~$3$.
\end{theorem}
\begin{proof}
We reduce from the {\sc Not-All-Equal Positive $3$-Satisfiability} problem, which is \NP-complete~\cite{Sc78}.
Let $(X,{\cal C})$ be a positive $3$-formula, where $X= \{x_1,x_2,\ldots,x_n\}$ and ${\cal C} = \{C_1, C_2,\ldots, C_m\}$.
From $(X,{\cal C})$ we construct a coloured graph~$(G,c)$ as follows.
The idea is to build the graph~$G$ as a sequence of gadgets, where two consecutive gadgets are joined by two edges.

We first construct the variable gadgets.
For every~$x_i$, let~$\ell_i$ be the number of clauses in which~$x_i$ appears.
We construct the gadget displayed in \figurename~\ref{f-variable}.
In particular, we introduce the new colours $a_i,b_i,c_i,d_i,\alpha_i^1,\ldots,\alpha_i^{2\ell_i}$.

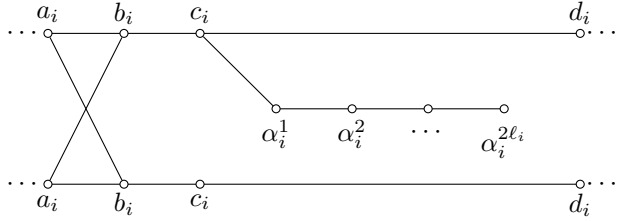
\begin{figure}
\begin{tikzpicture}
\vargadgetPathwidth
\end{tikzpicture}
\caption{\label{f-variable}The variable gadget for Theorem~\ref{t-pathwidth}.}
\end{figure}

We now construct the clause gadgets.
Let $C_j=\{x_g,x_h,x_i\}$.
Suppose that~$C_j$ is the $r$th clause in which~$x_g$ occurs, the $s$th clause in which~$x_h$ occurs, and the $t$th clause in which~$x_i$ occurs.
We construct the gadget displayed in \figurename~\ref{f-clause}.
In particular, we select the already introduced colours $\alpha_g^{2r-1}$, $\alpha_g^{2r}$, $\alpha_h^{2s}$ and~$\alpha_i^{2t}$, and also introduce six new colours $e_j,f_j,g_j,h_j,i_j,\beta_j$.

\begin{figure}
\begin{tikzpicture}
\clausegadgetPathwidth
\end{tikzpicture}
\caption{\label{f-clause}The clause gadget for Theorem~\ref{t-pathwidth}.}
\end{figure}
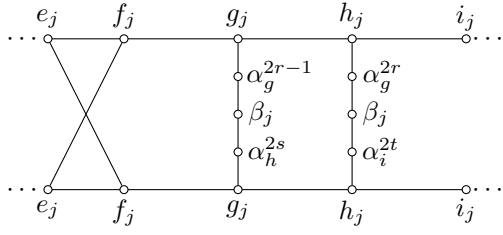

We obtain~$G$ by first adding for $l\in\{1,\ldots,n-1\}$, an edge between the top vertices coloured~$d_l$ and~$a_{l+1}$ and an edge between the bottom vertices coloured~$d_l$ and~$a_{l+1}$.
In this way we obtained a sequence of variable gadgets.
We now add an edge between the top vertices coloured~$d_n$ and~$e_1$ and the bottom vertices coloured~$d_n$ and~$e_1$.
This is followed by edges between the top vertices coloured~$i_l$ and~$e_{l+1}$ and an edge between the bottom vertices coloured~$i_l$ and~$e_{l+1}$ for $l\in \{1,\ldots,m-1\}$.
This reduction is illustrated in \figurename~\ref{fig:full-graph-small-pathwidth}.
Note that~$G$ is planar and bipartite, and that it admits a path decomposition of width~$3$, as shown in \figurename~\ref{fig:full-graph-small-pathwidth}.

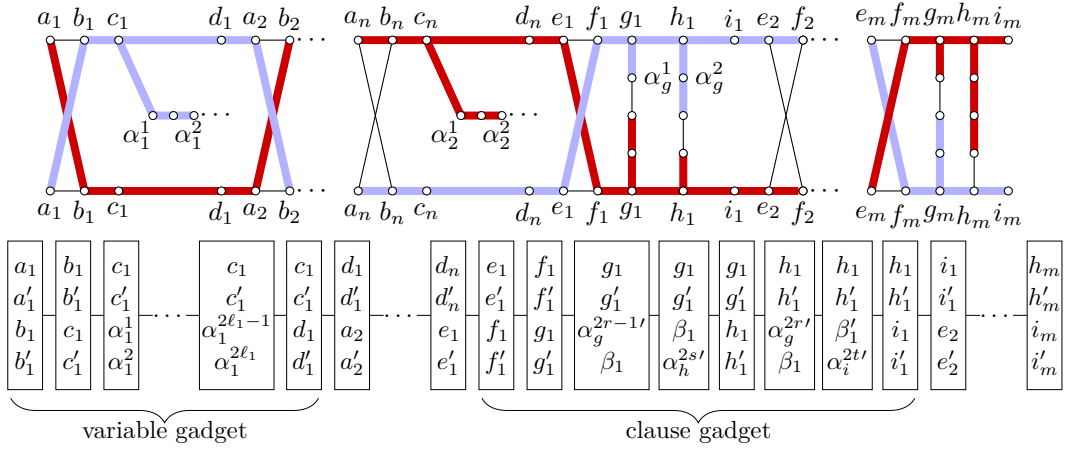
\begin{figure}
	\centering
	\begin{tikzpicture} [xscale=.45]
	   %%%%left part

	   \path [] (0,0)
	     node[minimum size=0pt] (Za1) {}
	   -- ++(1,0) node[minimum size=0pt] (Zb1)   {}
	   -- ++(1,0) node[minimum size=0pt] (Zc1)  {}
	   -- ++(3,0) node[minimum size=0pt] ()  {}
	   --++(1,0) node[minimum size=0pt] (Za2) {}
	   -- ++(1,0) node[minimum size=0pt] (Zb2)   {}
	   ;

	   	   \path [] (0,-2)
	   node[minimum size=0pt] (Za1p) {}
	   -- ++(1,0) node[minimum size=0pt] (Zb1p)   {}
	   -- ++(1,0) node[minimum size=0pt] ()  {}
	   -- ++(3,0) node[minimum size=0pt] ()  {}
	   --++(1,0) node[minimum size=0pt] (Za2p) {}
	   -- ++(1,0) node[minimum size=0pt] (Zb2p)   {}
	   ;
	   \path(Za1)--(Zb1p);
	   \path(Za1p)--(Zb1);
	   \path(Za2)--(Zb2p);
	   \path(Za2p)--(Zb2);
	   \path  (Zc1)
	   -- ++(1,-1) node[minimum size=0pt] (Zx1)  {}
	   -- ++(.6,0) node[minimum size=0pt]  {}
	   -- ++(.6,0) node[minimum size=0pt] (Zxl)   {};

	   \begin{scope}[comp1]
	   \path (Za1)--(Zb1p)--(Za2p)--(Zb2);
	   \end{scope}
	   \begin{scope}[comp2]
	   	 \path (Za1p)--(Zb1)-- (Za2)--(Zb2p);
	   	 \path (Zc1) -- (Zx1)--(Zxl);
	   \end{scope}

	   \path [] (0,0)
	     node (a1)[label=above:$a_1$, %label=left:$\cdots$
	   ]{}
	    ++(1,0) node (b1)  [label=above:$b_1$] {}
	    ++(1,0) node (c1) [label=above:$c_1$] {}
	    ++(3,0) node () [label=above:$d_1$] {}
	   ++(1,0) node (a2)[label=above:$a_2$]{}
	    ++(1,0) node (b2)  [label=right:$\cdots$, label=above:$b_2$] {}
	   ;

	   	   \path [] (0,-2)
	   node (a1p)[label=below:$a_1$, %label=left:$\cdots$
	   ]{}
	    ++(1,0) node (b1p)  [label=below:$b_1$] {}
	    ++(1,0) node () [label=below:$c_1$] {}
	    ++(3,0) node () [label=below:$d_1$] {}
	   ++(1,0) node (a2p)[label=below:$a_2$, %label=left:$\cdots$
	   ]{}
	    ++(1,0) node (b2p)  [label=right:$\cdots$, label=below:$b_2$] {}
	   ;
	   \path  (c1)
	    ++(1,-1) node (x1) [label=-100:$\alpha_1^1$] {}
	    ++(.6,0) node [label=-80:$\alpha_1^2$] {}
	    ++(.6,0) node (xl) [label=right:$\cdots$]  {};

	   %%%middle part
	   \path [] (9,0)
	   node[minimum size=0pt]  (Za1)   {}
	   -- ++(1,0) node[minimum size=0pt] (Zb1)   [] {}
	   -- ++(1,0) node[minimum size=0pt] (Zc1)  {}
	   -- ++(3,0) node[minimum size=0pt] ()  {}
	   --++(1,0) node[minimum size=0pt] (Ze1){}
	   -- ++(1,0) node[minimum size=0pt] (Zf1)   {}
	   --++(1,0) node[minimum size=0pt] (Zg1){}
	   --++(1.5,0) node[minimum size=0pt] (Zh1){}
	   --++(1.5,0) node[minimum size=0pt] (Zi1){}
	   --++(1,0) node[minimum size=0pt] (Ze2){}    --++(1,0) node (Zf2) {}

	   ;

	   \path [] (9,-2)
	   node[minimum size=0pt]  (Za1p)   {}
	   -- ++(1,0) node[minimum size=0pt] (Zb1p)    {}
	   -- ++(1,0) node[minimum size=0pt] ()   {}
	   -- ++(3,0) node[minimum size=0pt] ()  {}
	   --++(1,0) node[minimum size=0pt] (Ze1p){}
	   -- ++(1,0) node[minimum size=0pt] (Zf1p)    {}
	   --++(1,0) node[minimum size=0pt] (Zg1p){}
	   --++(1.5,0) node[minimum size=0pt] (Zh1p){}
	   --++(1.5,0) node[minimum size=0pt] (Zi1p){}
	   --++(1,0) node[minimum size=0pt] (Ze2p){}    --++(1,0) node (Zf2p) {}
	   ;

	   \path(Za1)--(Zb1p);
	   \path(Za1p)--(Zb1);
	   %\path(Za2)--(Zb2p);
	   %\path(Za2p)--(Zb2);
	   \path(Ze1)--(Zf1p);
       \path(Ze1p)--(Zf1);
	   \path(Ze2)--(Zf2p);
       \path(Ze2p)--(Zf2);
	   \path  (Zc1)
	   -- ++(1,-1) node[minimum size=0pt] (Zx1)  {}
	   -- ++(.6,0) node[minimum size=0pt]  {}
	   -- ++(.6,0) node[minimum size=0pt] (Zxl)   {};
	   \path(Zg1)
	      	   -- ++(0,-.5) node[minimum size=0pt] (Zm1)  {}
	      	   -- ++(0,-.5) node[minimum size=0pt] (Zm3) {}
	      	   -- ++(0,-.5) node[minimum size=0pt] (Zm5) {}
	      	   -- (Zg1p);
	   \path(Zh1)
	      	   -- ++(0,-.5) node[minimum size=0pt](Zm2)  {}
	      	   -- ++(0,-.5) node[minimum size=0pt](Zm4)  {}
	      	   -- ++(0,-.5) node[minimum size=0pt](Zm6)  {}
	      	   -- (Zh1p);

	   \begin{scope}[comp1]
	   \path (Za1)--(Ze1)--(Zf1p)--(Zf2p);
	   \path (Zc1) -- (Zx1)--(Zxl);
	   \path (Zg1p)--(Zm3);
	   \path (Zh1p)--(Zm6);
	   \end{scope}
	   \begin{scope}[comp2]
	   \path (Za1p)--(Ze1p)--(Zf1)--(Zf2);
	    \path (Zg1)--(Zm1);
	    \path (Zh1)--(Zm4);
	   \end{scope}

	   \path [] (9,0)
	   node  (a1)  [label=above:$a_n$] {}
	    ++(1,0) node (b1) [label=above:$b_n$]  [] {}
	    ++(1,0) node (c1) [label=above:$c_n$] {}
	    ++(3,0) node () [label=above:$d_n$] {}
	   ++(1,0) node (e1)[label=above:$e_1$ ]{}
	    ++(1,0) node (f1)  [label=above:$f_1$ ] {}
	   ++(1,0) node (g1)[label=above:$g_1$ ]{}
	   ++(1.5,0) node (h1)[label=above:$h_1$ ]{}
	   ++(1.5,0) node (i1)[label=above:$i_1$ ]{}
	   ++(1,0) node (e2)[label=above:$e_2$ ]{}    ++(1,0) node (f2)[label=right:$\cdots$,     label=above:$f_2$ ]{}

	   ;

	   \path [] (9,-2)
	   node  (a1p)  [label=below:$a_n$] {}
	    ++(1,0) node (b1p)  [label=below:$b_n$] [] {}
	    ++(1,0) node ()  [label=below:$c_n$] {}
	    ++(3,0) node () [label=below:$d_n$] {}
	   ++(1,0) node (e1p)[label=below:$e_1$ ]{}
	    ++(1,0) node (f1p)   [label=below:$f_1$ ] {}
	   ++(1,0) node (g1p)[label=below:$g_1$ ]{}
	   ++(1.5,0) node (h1p)[label=below:$h_1$ ]{}
	   ++(1.5,0) node (i1p)[label=below:$i_1$ ]{}
	   ++(1,0) node (e2p)[label=below:$e_2$ ]{}    ++(1,0) node (f2p)[label=right:$\cdots$, label=below:$f_2$ ] {}
	   ;

	   \path  (c1)
	   ++(1,-1) node (x1) [label=-100:$\alpha_2^1$] {}
	   ++(.6,0) node [label=-80:$\alpha_2^2$] {}
	   ++(.6,0) node (xl) [label=right:$\cdots$]  {};
	   \path(g1)
	      	   ++(0,-.5) node (m1) [label=right:$\alpha_g^1$] {}
	      	   ++(0,-.5) node (m3) {}
	      	   ++(0,-.5) node (m5) {}
	      	   (g1p);
	   \path(h1)
	      	   ++(0,-.5) node(m2) [label=right:$\alpha_g^2$] {}
	      	   ++(0,-.5) node(m4)  {}
	      	   ++(0,-.5) node(m6)  {}
	      	   (h1p);

	      %%%right part

	   \path [] (24,0)
	   node[minimum size=0pt]  (Ze1){}
	   -- ++(1,0) node[minimum size=0pt] (Zf1)   {}
	   --++(1,0) node[minimum size=0pt] (Zg1){}
	   --++(1,0) node[minimum size=0pt] (Zh1){}
	   --++(1,0) node[minimum size=0pt] (Zi1){}
	   ;
	   \path [] (24,-2) node[minimum size=0pt] (Ze1p){}
	   -- ++(1,0) node[minimum size=0pt] (Zf1p)   {}
	   --++(1,0) node[minimum size=0pt] (Zg1p){}
	   --++(1,0) node[minimum size=0pt] (Zh1p){}
	   --++(1,0) node[minimum size=0pt] (Zi1p){}
	   ;
	   \path(Ze1)--(Zf1p);
	   \path(Ze1p)--(Zf1);
	     \path(Zg1)
	   -- ++(0,-.5) node[minimum size=0pt] (Zm1) {}
	   -- ++(0,-.5) node[minimum size=0pt] (Zm3) {}
	   -- ++(0,-.5) node[minimum size=0pt] (Zm5) {}
	   -- (Zg1p);
	   \path(Zh1)
	   -- ++(0,-.5) node[minimum size=0pt](Zm2)  {}
	   -- ++(0,-.5) node[minimum size=0pt](Zm4)  {}
	   -- ++(0,-.5) node[minimum size=0pt](Zm6)  {}
	   -- (Zh1p);

	   \begin{scope}[comp2]
	   \path (Ze1)--(Zf1p)--(Zi1p);
	   \path (Zg1p)--(Zm3);
	   \end{scope}
	   \begin{scope}[comp1]
	   \path (Ze1p)--(Zf1)--(Zi1);
	   \path (Zg1)--(Zm1);
	   \path (Zh1)--(Zm6);
	   \end{scope}

	   \path [] (24,0)
	   node  (e1)[label=above:$e_m$ ]{}
	    ++(1,0) node (f1)  [label=above:$f_m$   ] {}
	   ++(1,0) node (g1)[label=above:$g_m$ ]{}
	   ++(1,0) node (h1)[label=above:$h_m$ ]{}
	   ++(1,0) node (i1)[label=above:$i_m$ ]{}
	   ;
	   \path [] (24,-2) node (e1p)[label=below:$e_m$ ]{}
	    ++(1,0) node (f1p)  [	  label=below:$f_m$  ] {}
	   ++(1,0) node (g1p)[label=below:$g_m$ ]{}
	   ++(1,0) node (h1p)[label=below:$h_m$ ]{}
	   ++(1,0) node (i1p)[label=below:$i_m$ ]{}
	   ;
	     \path(g1)
	    ++(0,-.5) node (m1) {}
	    ++(0,-.5) node (m3) {}
	    ++(0,-.5) node (m5) {}
	    (g1p);
	   \path(h1)
	    ++(0,-.5) node(m2)  {}
	    ++(0,-.5) node(m4)  {}
	    ++(0,-.5) node(m6)  {}
	    (h1p);
	\end{tikzpicture}

	\newcommand{\bag}[1]{
		node [align=center, text width=14pt,text height=3ex,text depth=11ex,rectangle] {#1}
	}
	\newcommand{\largebag}[1]{
		node [align=center, text width=20pt,text height=3ex,text depth=11ex,rectangle] {#1}
	}
	\newcommand{\vlargebag}[1]{
		node [align=center, text width=30pt,text height=3ex,text depth=11ex,rectangle] {#1}
	}

        \scalebox{0.93}{
	\begin{tikzpicture}[xscale=.68]
	\path (0,0) \bag{$a_1$ $a'_1$ $b_1$ $b_1'$}
	--++(1,0) \bag{$b_1$ $b'_1$ $c_1$ $c_1'$}
	--++(1,0) \bag{$c_1$ $c_1'$ $\alpha_1^1$ $\alpha_1^2$}
	--++(1,0) node[draw=none] {\ldots}
	--++(1.4,0) \vlargebag{$c_1$\\ $c_1'$ $\alpha_1^{{2\ell_1{-}1}}$ $\alpha_1^{2\ell_1}$}
	--++(1.4,0) \bag{$c_1$ $c_1'$ $d_1$ $d_1'$}
	--++(1,0) \bag{$d_1$ $d_1'$ $a_2$ $a_2'$}
	--++(1,0) node[draw=none] {\ldots}
	--++(1,0) \bag{$d_n$ $d_n'$ $e_1$ $e_1'$}
	--++(1,0) \bag{$e_1$ $e_1'$ $f_1$ $f_1'$}
	--++(1,0) \bag{$f_1$ $f_1'$ $g_1$ $g_1'$}
	--++(1.4,0) \vlargebag{$g_1$\\ $g_1'$ $\alpha_g^{2r{-}1}{'}$ $\beta_1$}
	--++(1.5,0) \largebag{$g_1$ $g_1'$ $\beta_1$ $\alpha_h^{2s}{'}$}
	--++(1.1,0) \bag{$g_1$ $g_1'$ $h_1$ $h_1'$}
	--++(1.1,0) \largebag{$h_1$\\ $h_1'$ $\alpha_g^{2r}{'}$ $\beta_1$}
	--++(1.2,0) \largebag{$h_1$ $h_1'$ $\beta_1'$ $\alpha_i^{2t}{'}$}
	--++(1.1,0) \bag{$h_1$ $h_1'$ $i_1$ $i_1'$}
	--++(1,0) \bag{$i_1$ $i_1'$ $e_2$ $e_2'$}
	--++(1,0) node[draw=none] {\ldots}
	--++(1,0) \bag{$h_m$ $h_m'$ $i_m$ $i_m'$}
	;
	\draw [decorate,decoration={brace,amplitude=10pt},xshift=0pt,yshift=-7.8ex]
	(6.1,0) -- (-.3,0) node[rectangle,midway, fill=none, draw=none, yshift=-3.2ex ] {variable gadget} ;
	\draw [decorate,decoration={brace,amplitude=10pt},xshift=0pt,yshift=-7.8ex]
	(18.5,0) -- (9.5,0) node[rectangle, midway, fill=none, draw=none, yshift=-3.2ex ] {clause gadget} ;

	\end{tikzpicture}}
	\caption{\label{fig:full-graph-small-pathwidth}Top: The graph~$G$ obtained by the reduction from {\sc Not-All-Equal Positive $3$-Satisfiability}, with a partition into red and blue components, corresponding to a truth assignment.
The graph~$G$ can be seen to be planar by letting one of the crossing edges~$\{a_i,b_i\}$ in each variable gadget and one of the crossing edges~$\{e_i,f_i\}$ in each clause gadget go round the outside of the rest of~$G$.
Bottom: A path-decomposition of~$G$ of width~$3$ (vertices are identified by their colour, we use a~$'$ for the bottom (right) vertex in order to distinguish it from the top (left) vertex of the same colour).
}
\end{figure}

We claim that $(X,{\cal C})$ has a satisfying truth assignment if and only if~$(G,c)$ has a colourful partition of size at most~$2$.

First suppose that~$(X,{\cal C})$ has a satisfying truth assignment~$\tau$.
We will define a colourful partition of~$(G,c)$ with two colourful sets, which we call \textcolor{colour1!80!black}{\em red} and \textcolor{colour2!80!black}{\em blue}, respectively.
Vertices of variable gadgets will be allocated to exactly one of these two sets based on the truth value of the variable.
Similarly, vertices of clause gadgets will be allocated to exactly one of these sets based on the truth values of their literals.
Below we describe all possibilities.
It is readily seen that both the blue set and the red set are colourful.
Dotted lines indicate two options depending on the previous gadget in the sequence.
We choose the option that ensures the red and blue sets remain connected.
\begin{itemize}
\item $\tau(x_i)=\True$:\\
\begin{tikzpicture}
\vargadgetPathwidthnonodes

\begin{scope}[comp1]
\path[] (Zb2)--(Zd2);
\path[maybe] (Za)--(Zb2)--(Za2);
\end{scope}
\begin{scope}[comp2]
\path[] (Zb)--(Zd);
\path[] (Zc)--(Zx1)--(Zxl);
\path[maybe] (Za)--(Zb)--(Za2);
\end{scope}
\vargadgetPathwidthnolines
\end{tikzpicture}
\item $\tau(x_i)=\False$:\\
\begin{tikzpicture}
\vargadgetPathwidthnonodes
\begin{scope}[comp1]
\path[] (Zb)--(Zd);
\path[] (Zc)--(Zx1)--(Zxl);
\path[maybe] (Za)--(Zb)--(Za2);
\end{scope}

\begin{scope}[comp2]
\path[] (Zb2)--(Zd2);
\path[maybe] (Za2)--(Zb2)--(Za);
\end{scope}
\vargadgetPathwidthnolines

\end{tikzpicture}

\item The clause $C_j=\{x_g,x_h,x_i\}$, with $\tau(x_g)=\tau(x_h)=\True$ and $\tau(x_i)=\False$ (exchange the colours blue and red for opposite truth values):\\
\begin{tikzpicture}
\clausegadgetPathwidthnonodes

\begin{scope}[comp1]
\path (Zb)--(Zg);
\path (Zy)--(Zc);
\path (Zx2)--(Zf);
\path[maybe] (Za)--(Zb)--(Za2);
\end{scope}

\begin{scope}[comp2]
\path (Zb2)--(Zg2);
\path (ZC2)--(Zf2);
\path[maybe] (Za2)--(Zb2)--(Za);
\end{scope}
\clausegadgetPathwidthnolines
\end{tikzpicture}

\item The clause $C_j=\{x_g,x_h,x_i\}$, with $\tau(x_g)=\True$ and $\tau(x_h)=\tau(x_i)=\False$ (exchange the colours blue and red for opposite truth values).\\
\begin{tikzpicture}
\clausegadgetPathwidthnonodes

\begin{scope}[comp1]
\path (Zb)--(Zg);
\path (ZC1)--(Zc);
\path (Zx2)--(Zf);
\path[maybe] (Za)--(Zb)--(Za2);
\end{scope}

\begin{scope}[comp2]
\path (Zb2)--(Zg2);
\path (ZC2)--(Zf2);
\path (Zy)--(Zc2);
\path[maybe] (Za2)--(Zb2)--(Za);
\end{scope}
\clausegadgetPathwidthnolines
\end{tikzpicture}

\item The clause $C_j=\{x_g,x_h,x_i\}$, with $\tau(x_g)=\tau(x_i)=\True$ and $\tau(x_h)=\False$ (exchange the colours blue and red for opposite truth values).\\
\begin{tikzpicture}
\clausegadgetPathwidthnonodes

\begin{scope}[comp1]
\path (Zb)--(Zg);
\path (Zz)--(Zf);
\path (Zx1)--(Zc);
\path[maybe] (Za)--(Zb)--(Za2);
\end{scope}

\begin{scope}[comp2]
\path (Zb2)--(Zg2);
\path (ZC1)--(Zc2);
\path[maybe] (Za2)--(Zb2)--(Za);
\end{scope}
\clausegadgetPathwidthnolines
\end{tikzpicture}

\end{itemize}

Now suppose that~$(G,c)$ has a colourful partition~$(V_1,V_2)$.
We define a truth assignment~$\tau$ as follows.
In the gadget for the variable~$x_i$, the vertices with colours $\alpha_i^1,\ldots,\alpha_i^{2\ell_i}$ belong to the same component (otherwise, by connectedness, one component consists of only a subset of these vertices, so the other component contains two vertices with the same colour, a contradiction).
We define $\tau(x_i)=\True$ if they belong to~$V_1$ and $\tau(x_i)=\False$ if they belong to~$V_2$.
In the gadget for the clause $C_j=\{x_g,x_h,x_i\}$, there are two vertices with the same colour~$\beta_j^1$, so they must belong to different sets~$V_1$ and~$V_2$.
These two vertices are only connected to vertices with colours $\alpha_g^{2r-1}$, $\alpha_g^{2r}$, $\alpha_h^{2s}$ and~$\alpha_i^{2t}$.
Hence, as~$V_1$ and~$V_2$ are connected, one of these four vertices must be in~$V_1$ and another one must be in~$V_2$.
This means that some variable in~$C_j$ was set to $\True$ and another one was set to $\False$.
We conclude that~$\tau$ is a satisfying truth assignment.
\end{proof}

The idea behind our next result is that the sets~$V_1$ and~$V_2$ of a colourful partition of size~$2$ form connected subtrees in a tree decomposition.
Branching over all options, we ``guess'' two vertices~$a$ and~$b$ of one bag to belong to different sets~$V_i$.
By exploiting the treewidth-$2$ assumption we can translate the instance into an equivalent instance of {\sc $2$-Satisfiability}.

\begin{theorem}\label{t-treewidth2}
{\sc $2$-Colourful Partition} is polynomial-time solvable for graphs of treewidth at most~$2$.
\end{theorem}

\begin{proof}
Let~$(G,c)$ be a coloured graph on~$n$ vertices such that~$G$ has treewidth at most~$2$.
Without loss of generality we may assume that~$G$ is connected.
We may assume that~$G$ is not colourful, otherwise we are trivially done.
If~$G$ has treewidth~$1$, then it is a tree and we apply Theorem~\ref{t-trees-fpt-components}.
Hence we may assume that~$G$ has treewidth~$2$.
Let~$(T,{\cal X})$ be a tree decomposition of~$G$ of width~$2$ (so all bags of~${\cal X}$ have size at most~$3$).
We can obtain this tree decomposition in linear time~\cite{WC83}.
Let us state and then explain two properties that we may assume hold for $(T,{\cal X})$.
These properties, along with the main definitions necessary for our algorithm, are illustrated in \figurename~\ref{f-exex}.

\begin{enumerate}
\item\label{property:one}For any two adjacent nodes $i,j$ in~$T$, one of~$X_i$ and~$X_j$ strictly contains the other.
\item\label{property:two}All bags are pairwise distinct.
\end{enumerate}

\renewcommand{\floatpagefraction}{.8}%  

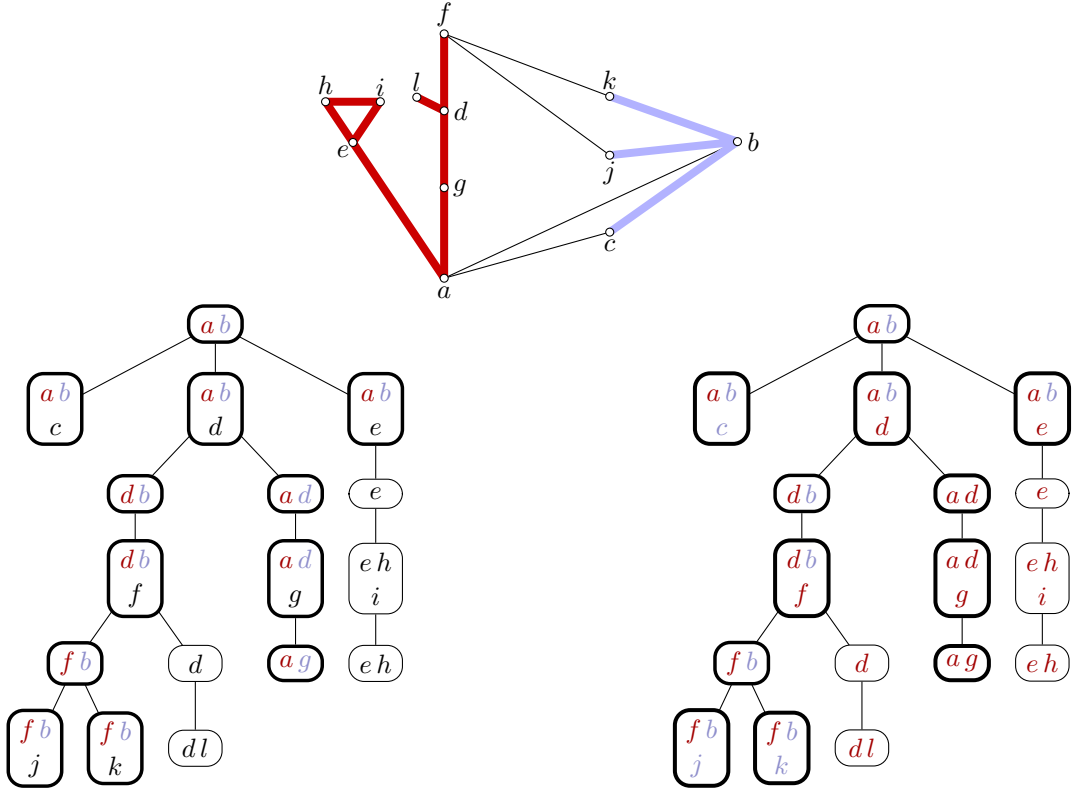
\begin{figure}
	\centering

	\begin{tikzpicture}[scale=1.2, rotate=90]
\begin{scope}[yshift=63,xshift=142]
\begin{scope}[rotate=90]
	\node [minimum size=0pt] (Za) at (1,5.5) {} ;
	\node [minimum size=0pt] (Zd) at (1,3.65) {} ;
	\node [minimum size=0pt] (Ze) at (2,4) {} ;
	\node [minimum size=0pt] (Zf) at (1,2.8) {} ;
	\node [minimum size=0pt] (Zg) at (1,4.5) {} ;
	\node [minimum size=0pt] (Zh) at (2.3,3.55) {} ;
	\node [minimum size=0pt] (Zi) at (1.7,3.55) {} ;
	\node [minimum size=0pt] (Zl) at (1.3,3.5) {} ;
\end{scope}
\end{scope}
	\node [minimum size=0pt] (Zb) at (1,0) {} ;
	\node [minimum size=0pt] (Zc) at (0,1.4) {} ;
	\node [minimum size=0pt] (Zj) at (.85,1.4) {} ;
	\node [minimum size=0pt] (Zk) at (1.5,1.4) {} ;

	\path   (Za) -- (Zc)-- (Zb)-- (Zj)-- (Zf)-- (Zk)-- (Zb);
	\path (Ze)-- (Zh)-- (Zi)-- (Ze)--(Za)-- (Zg)-- (Zd)-- (Zl);
	\path (Zd)-- (Zf);
	\path (Za)-- (Zb);
	\begin{scope}[comp1]
	\path   (Za) -- (Zg)-- (Zd)-- (Zf);
	\path   (Za) -- (Ze)-- (Zh)-- (Zi) --(Ze);
	\path (Zd)-- (Zl);
	\end{scope}
	\begin{scope}[comp2]
	\path   (Zk) -- (Zb)-- (Zj);
	\path (Zb)-- (Zc);
	\end{scope}

\begin{scope}[yshift=63,xshift=142]
\begin{scope}[rotate=90]
	\node [label=below:$a$] (a) at (1,5.5) {} ;
	\node [label=right:$d$] (d) at (1,3.65) {} ;
	\node [label=220:$e$] (e) at (2,4) {} ;
	\node [label=above:$f$] (f) at (1,2.8) {} ;
	\node [label=right:$g$] (g) at (1,4.5) {} ;
	\node [label=above:$h$] (h) at (2.3,3.55) {} ;
	\node [label=above:$i$] (i) at (1.7,3.55) {} ;
	\node [label=above:$l$] (l) at (1.3,3.5) {} ;
\end{scope}
\end{scope}
	\node [label=right:$b$] (b) at (1,0) {} ;
	\node [label=below:$c$] (c) at (0,1.4) {} ;
	\node [label=below:$j$] (j) at (.85,1.4) {} ;
	\node [label=above:$k$] (k) at (1.5,1.4) {} ;
	\end{tikzpicture}

	\newcommand{\lowNode}[3] {
		node []
		{$#1\ifx\hfuzz#2\hfuzz{}\else\, #2\fi$\ifx\hfuzz#3\hfuzz{}\else \\ $#3$ \fi}
	}
	\newcommand{\cutNode}[3] {
		node [line width=1.3pt]
		{$\textcolor{colour1!80!black}{#1} \textcolor{colour2!80!black}{\,#2}$\ifx\hfuzz#3\hfuzz{}\else \\ $#3$ \fi}
	}

	\newcommand{\cutNodeR}[3] {
		\cutNode{#1}{#2}{#3}
	}
	\newcommand{\cutNodeB}[3] {
		\cutNode{#1}{#2}{#3}
	}
	\newcommand{\cutNodeb}[3] {
		\cutNode{#1}{#2}{#3}
	}

	\newcommand{\cutNoder}[3] {
		\cutNode{#1}{#2}{#3}
	}
	\newcommand{\niceTree}{%
		\begin{tikzpicture}[
		scale=.75,
		every node/.style={shape=rectangle, draw, align=center, minimum width=2em, rounded corners=2mm},
		sibling distance=3em,
		level 1/.style={sibling distance=8em},
		level 2/.style={sibling distance=8em},
		level 3/.style={sibling distance=8em},
		level 4/.style={sibling distance=6em},
		level 5/.style={sibling distance=4em}]

		\node[draw=none] {}
		\cutNode ab{}
		child[shorten >=-1pt,shorten <=-2pt] {\cutNodeb abc }
		child {\cutNoder abd
			child[shorten >=-1pt,shorten <=-2pt] {  \cutNode db{}
				child[shorten >=0pt,shorten <=0pt] {  \cutNoder dbf
					child[shorten >=-1pt,shorten <=-3pt] {  \cutNode fb{}
						child[shorten >=-1pt,shorten <=-1pt] {  \cutNodeB fbj }
						child[shorten >=-1pt,shorten <=-1pt] {  \cutNodeb fbk }
					}
					child[shorten >=-0.3pt,shorten <=-3pt] {\lowNode d{}{}
						child[shorten >=0pt,shorten <=0pt] {\lowNode dl{}}
					}
				}
			}
			child[shorten >=-1pt,shorten <=-2pt] {  \cutNodeR ad{}
				child[shorten >=0pt,shorten <=0pt] {  \cutNodeR adg
					child {  \cutNodeR ag{} }
				}
			}
		}
		child[shorten >=-1pt,shorten <=-2pt] { \cutNoder abe
			child[shorten >=0pt,shorten <=0pt] {\lowNode e{}{}
				child {\lowNode ehi
					child {\lowNode eh{}}
				}
			}
		};
		\end{tikzpicture}%
	}
	\niceTree%
	\renewcommand{\lowNode}[3] {
		node []
		{\textcolor{colour1!80!black}{$#1\ifx\hfuzz#2\hfuzz{}\else\, #2\fi$}\ifx\hfuzz#3\hfuzz{}\else \\ \textcolor{colour1!80!black}{$#3$} \fi}
	}%
	\renewcommand{\cutNodeR}[3] {
		node [line width=1.5pt]
		{\textcolor{colour1!80!black}{$#1\,#2$}\ifx\hfuzz#3\hfuzz{}\else \\ \textcolor{colour1!80!black}{$#3$} \fi}
	}%
	\renewcommand{\cutNodeB}[3] {
		node [line width=1.5pt]
		{\textcolor{colour1!80!black}{$#1$}\,\textcolor{colour2!80!black}{$#2$}\ifx\hfuzz#3\hfuzz{}\else \\ \textcolor{colour2!80!black}{$#3$ }\fi}
	}%
	\renewcommand{\cutNodeb}[3] {
		node [line width=1.5pt]
		{$\textcolor{colour1!80!black}{#1} \textcolor{colour2!80!black}{\,#2}$\ifx\hfuzz#3\hfuzz{}\else \\  \textcolor{colour2!80!black}{$#3$} \fi}
	}%
	\renewcommand{\cutNoder}[3] {
		node [line width=1.5pt]
		{$\textcolor{colour1!80!black}{#1} \textcolor{colour2!80!black}{\,#2}$\ifx\hfuzz#3\hfuzz{}\else \\ \textcolor{colour1!80!black}{$#3$} \fi}
	}%
	\hfill\niceTree
	\caption{\label{f-exex}Illustration of the polynomial-time algorithm for treewidth~$2$.
Top: an input graph (the vertex colouring is not represented), with a partition cutting~$(a,b)$, i.e.~$a$ is in one part, $b$ is in the other.
Bottom-left: a tree-decomposition of the graph satisfying Properties~\ref{property:one}, \ref{property:two} and~\ref{property:three}.
The bags in the head subtree are in bold.
Precuts computed as in Claim~\ref{lem:head-subtree-is-precut} are shown using colours on the top line of bags that have precuts.
For example, the bag $\{d,b,f\}$ is precut on~$(d,b)$, as can be verified directly in the graph (if~$d$,~$b$ and~$f$ are not all in the same part of an $(a,b)$-partition, then~$d$ must be in the same part as~$a$).
Here~$f$ is attached to~$b$, and~$e$ is attached to~$a$ but not to~$b$.
Bottom-right: the tree decomposition showing a possible output of the {\sc $2$-Satisfiability} formula, where the colours represent the values of each~$x_u$.}
\end{figure}

\noindent
In fact, what we will show is that if we find that $(T,{\cal X})$ does not have these two properties, then we can make simple changes to obtain a tree decomposition that does.
First, if there are two adjacent nodes~$i,j$ in~$T$ such that neither~$X_i$ nor~$X_j$ contains the other, we remove the edge from~$ij$ from~$T$, create a new node~$k$ that is adjacent to both~$i$ and~$j$ and let $X_k=X_i\cap X_j$.
Note that $X_k \neq \emptyset$, as~$G$ is connected.
We now have a tree decomposition that has Property~\ref{property:one} unless $X_i=X_j$ for a pair of adjacent nodes~$i$ and~$j$.
But now for \emph{every} pair of identical bags~$X_i$ and~$X_j$, we delete~$j$ and make each of its neighbours adjacent to~$i$ and so obtain a tree decomposition with Properties~\ref{property:one} and~\ref{property:two}.

\medskip
\noindent Let~$a$ and~$b$ be a fixed pair of adjacent vertices in~$G$.
Almost all of the remainder of this proof is concerned with describing an algorithm that decides whether or not~$G$ has a colourful partition of size~$2$ in which~$a$ and~$b$ belong to different parts.
Clearly such an algorithm suffices: we can apply it to each of the~$O(n^2)$ pairs of adjacent vertices in~$G$ to determine whether~$G$ has \emph{any} colourful partition of size~$2$ (since such a partition must separate at least one pair of adjacent vertices).

As~$a$ and~$b$ are adjacent, there exists at least one bag that contains both of them.
We may assume that this bag is $X_0=\{a,b\}$; otherwise we simply add a new node~$0$ to~$T$ with~$X_0=\{a,b\}$ and make~$0$ adjacent to~$i$ such that~$X_i$ is a (larger) bag that contains both~$a$ and~$b$.

We now orient all edges of~$T$ away from~$0$ and think of~$T$ as being rooted at~$0$.
We write $i\rightarrow j$ to denote an edge oriented from~$i$ to~$j$.
If $i\rightarrow j$ is present in~$T$, then~$i$ is the {\em parent} of~$j$ and~$j$ is a {\em child} of~$i$.
The \emph{head subtree} of~$T$ is the subtree obtained by removing all $1$-vertex bags along with all their descendants.
For any oriented edge $i\rightarrow j$ in~$T$, we have either~$X_i \subsetneq X_j$ or~$X_j\subsetneq X_i$ by Property~\ref{property:one}.
The tree~$T[i]$ denotes the subtree of~$T$ rooted at~$i$; in particular $T=T[0]$.
The set~$V[i]$ denotes $\bigcup_{j \in V(T[i])} X_j$.
Finally, as we shall explain, we may assume the following property.

\begin{enumerate}
\setcounter{enumi}{2}
\item\label{property:three}For any node~$i$ of $T$, the subgraph~$G[V[i]]$ induced by~$V[i]$ is connected.
\end{enumerate}

\noindent
It is possible that we must again modify the tree decomposition to obtain this property.
Suppose that it does not hold for some node~$i$.
That is, the vertices of~$V[i]$ can be divided into two sets~$U$ and~$W$ such that there is no edge from~$U$ to~$W$ in~$G$.
We create two trees~$T_{i,U}$ and~$T_{i,W}$ that are isomorphic to~$T[i]$: for each vertex~$j$ in~$T[i]$, we let~$j_U$ and~$j_W$ be the corresponding nodes in~$T_{i,U}$ and~$T_{i,W}$ respectively and let $X_{j_U} = X_j \cap U$, $X_{j_W} = X_j \cap W$.
Then the tree decomposition is modified by replacing~$T[i]$ by~$T_{i,U}$ and~$T_{i,W}$ and making each of~$i_U$ and~$i_W$ adjacent to the parent of~$i$.
(The vertex~$i$ certainly has a parent since $i=0$ would imply that~$G$ is not connected.)
If at any point we create a node whose associated bag is empty  or identical to that of its parent, we delete it and make its children (if it has any) adjacent to its parent.
In this way we obtain a decomposition that now satisfies each of Properties~\ref{property:one}, \ref{property:two} and~\ref{property:three}.

\medskip
\noindent
We say that a colourful partition $P=(V_1,V_2)$ of~$G$ is an {\em $(a,b)$-partition} if $a\in V_1$ and $b\in V_2$ and we say that~$P$ \emph{cuts} a pair~$(u_1,u_2)$ if
$u_1\in V_1$ and $u_2\in V_2$.
We emphasize that the order is important.
A colourful partition~$P$ \emph{respects} a bag~$X$ if~$X\subseteq V_1$ or~$X\subseteq V_2$.
Note that every colourful partition respects all $1$-vertex bags.
Let~$X$ be a bag that contains vertices~$u$ and~$v$.
Then~$X$ is \emph{precut} on~$(u,v)$ if every $(a,b)$-partition either cuts~$(u,v)$ or respects~$X$.
Note that~$X_0$ is precut on~$(a,b)$ by definition.
We make three structural claims.

\clm{\label{lem:join-hereditary}If an $(a,b)$-partition~$P$ respects a bag~$X_i$, then it respects every bag~$X_j$ such that $i\rightarrow j$.}

\medskip
\noindent
We prove Claim~\ref{lem:join-hereditary} as follows.
Consider the set~$\mathcal C$ of bags that are not respected by~$P$.
Then~$\mathcal C$ is the intersection of the set of bags containing at least one vertex from~$V_1$ and the set of bags containing at least one vertex from~$V_2$.
Since both~$V_1$ and~$V_2$ are connected in~$G$, the set of nodes whose bags contain at least one vertex from~$V_1$ and the set of nodes whose bags contain at least one vertex from~$V_2$ induce subtrees of~$T$.
Hence their intersection, $\mathcal C$, is a set of nodes that also induce a subtree of~$T$.
Since~$(a,b)$ is not respected by~$P$, we know $0\in\mathcal C$.
This means that if~$i$ is not in~$\mathcal C$, then for every other vertex $j \in T[i]$, $j$ cannot be in~$\mathcal C$.\dia

\clm{\label{lem:precut-X2-X3}Let $i\rightarrow j$ be an oriented edge of~$T$ with $X_i=\{u,v\}$ and $|X_j|=3$.
If~$X_i$ is precut on~$(u,v)$, then~$X_j$ is precut on~$(u,v)$.}

\medskip
\noindent
We prove Claim~\ref{lem:precut-X2-X3} as follows.
By Property~\ref{property:one}, we know that $X_i\subsetneq X_j$, so~$X_j$ contains~$u$ and~$v$.
Suppose~$(u,v)$ is a precut on~$X_i$ and let~$P$ be an $(a,b)$-partition.
If~$P$ cuts~$(u,v)$ then we are done.
Otherwise, $P$ respects $X_i=\{u,v\}$, so by Claim~\ref{lem:join-hereditary}, $P$ also respects~$X_j$.
Thus~$(u,v)$ is a precut on~$X_j$.\dia

\clm{\label{lem:precut-X3-X2}Let $i\rightarrow j$ be an oriented edge of~$T$ with $X_i=\{u,v,w\}$ and $|X_j|=2$ such that~$X_i$ is precut on~$(u,v)$.
If $X_j=\{v,w\}$, then~$X_j$ is precut on~$(w,v)$.
If $X_j=\{u,w\}$, then~$X_j$ is precut on~$(u,w)$.}

\medskip
\noindent
We prove Claim~\ref{lem:precut-X3-X2} as follows.
Let $P=(V_1,V_2)$ be an $(a,b)$-partition of~$G$.
Suppose~$P$ does not respect~$X_j$.
Then if~$Y$ contains~$w$, we must have that~$w$ is not in the same part of the partition as the other vertex --- either~$u$ or~$v$ --- of~$X_j$.
By Claim~\ref{lem:join-hereditary}, $P$ does not respect~$X_i$ either, so we know $u\in V_1$ and $v\in V_2$.
Thus if $X_j=\{v,w\}$, then $w \in V_1$, and if $X_j=\{u,w\}$, then $w \in V_2$.
\dia

\medskip
\noindent
By Claim~\ref{lem:join-hereditary} and the fact that $1$-vertex bags are respected, for every node~$i$ not in the head subtree of~$T$, $X_i$ is respected by every $(a,b)$-partition of~$G$. We now show the following claim.

\clm{\label{lem:head-subtree-is-precut}For every node~$i$ in the head subtree, $X_i$ is precut on some pair of its vertices,
and these precuts can be computed in linear time.}

\medskip
\noindent
We prove that Claim~\ref{lem:head-subtree-is-precut} holds by proving a slightly stronger statement: for all~$d$, for each node~$j$ in the head subtree at distance~$d$ from~$0$, $X_j$ is precut on a pair of its vertices, and if~$X_j$ contains three vertices, then it is precut on some pair~$(u,v)$ such that $X_i=\{u,v\}$ where~$i$ is the parent of~$j$.
We prove this by induction on~$d$.
The base case holds as~$X_0$ is precut on~$(a,b)$.
For the inductive case, suppose that~$j$ is some node at distance $d>0$ with parent~$i$.
As~$i$ and~$j$ are in the head subtree, each of~$X_i$ and~$X_j$ has either two or three vertices.
Suppose first that $X_j=\{u,v,w\}$, and then we can assume that~$X_i$ is, say, $\{u,v\}$.
(Recall that $X_i\subsetneq X_j$ or $X_i\subsetneq X_j$ by Property~\ref{property:one}.)
By the induction hypothesis, we know that~$X_i$ is precut on~$(u,v)$ or on~$(v,u)$ and so, by Claim~\ref{lem:precut-X2-X3}, $X_j$ is precut in the same way.
Now suppose that~$X_j$ has two vertices and $X_i=\{u,v,w\}$.
Since~$X_0$ has two vertices, it follows that $X_i \neq X_0$, so~$i$ has a parent~$h$, and we can suppose that $X_h=\{u,v\}$.
By the induction hypothesis, without loss of generality we may assume~$X_h$ is precut on~$(u,v)$, so~$X_i$ is precut on~$(u,v)$ by the same argument as above.
Then, as~${\cal X}$ has no identical bags by Property~\ref{property:two}, we find that $X_j=\{u,w\}$ or $X_j=\{v,w\}$.
It follows from Claim~\ref{lem:precut-X3-X2} that~$X_j$ is precut.~\dia

\medskip
\noindent
We say that two vertices~$u$ and~$v$ of~$G$ are \emph{attached} if they are adjacent or~${\cal X}$ contains the bag~$\{u,v\}$.
If~$u$ and~$v$ are not attached, they are \emph{detached}.

\clm{\label{lem:detached}Let $X_i=\{u,v,w\}$ be a $3$-vertex bag precut on~$(u,v)$.
If~$w$ and~$u$ (or~$w$ and~$v$) are detached, then in every $(a,b)$-partition of~$G$ the vertices~$w$ and~$v$ (respectively~$w$ and~$u$) are in the same colourful set.}

\medskip
\noindent
We prove Claim~\ref{lem:detached} as follows.
Let $P=(V_1,V_2)$ be an $(a,b)$-partition.
If~$P$ respects~$X_i$, then~$w$ is in the same colourful set as both~$u$ and~$v$.
Thus we may assume that~$P$ does not respect~$X_i$ and so $u\in V_1$ and $v\in V_2$, since~$X_i$ is precut on~$(u,v)$.
We may assume without loss of generality that~$w$ is detached from~$u$ and so we must show that $w\in V_2$.

Let us assume instead that $w\in V_1$ and derive a contradiction.
By the connectivity of~$G[V_1]$, there exists an induced path~$Q$ on~$\ell$ edges in~$G[V_1]$ that connects~$u$ to~$w$.
As~$u$ and~$w$ are detached, they are not adjacent, so~$\ell$ is at least~$2$.
No internal vertex of~$Q$ is in~$X_i$ since~$u$ and~$w$ are its end-vertices and~$v$ is not in~$V_1$.
Let $X_{i_1},\ldots,X_{i_\ell}$ be bags of~${\cal X}$ such that~$X_{i_j}$ is a bag that contains the pair of vertices joined by the $j$th edge of~$Q$.
We take a walk in~$T$ from~$i_1$ to~$i_\ell$ by stitching together paths from~$i_j$ to~$i_{j+1}$, $1 \leq j \leq \ell-1$.
As~$X_{i_j}$ and~$X_{i_{j+1}}$ correspond to incident edges of~$Q$, and both~$X_{i_j}$ and~$X_{i_{j+1}}$ contain the internal vertex of~$Q$ incident with both these edges, every bag of the nodes along the path between them also contains this vertex.
Thus every bag of the nodes along our walk from~$i_1$ to~$i_\ell$ contains an internal vertex of~$Q$ and so none of these bags is~$X_i$.
Therefore our walk must be contained within a connected component of $T \setminus i$.
Let~$j$ be the node of this component adjacent to~$i$ in~$T$.
We note that~$X_i$ contains~$u$ and~$w$, $X_{i_1}$ contains~$u$, $X_{i_\ell}$ contains~$w$, and the paths from~$i$ to each of~$i_1$ and~$i_\ell$ go through~$j$.
Thus~$X_j$ must contain~$u$ and~$w$, and so, by Property~\ref{property:one}, it follows that $X_j=\{u,w\}$.
As~$u$ and~$w$ are detached, we have our contradiction.\dia

\medskip
\noindent
We now build~$\phi$, an instance of {\sc $2$-Satisfiability}, that will help us to find an $(a,b)$-partition of~$G$ (if one exists).
For each vertex~$u$ of~$G$, we create a variable~$x_u$ (understood as ``$u\in V_1$'').
We add clauses (or pairs of clauses) on two variables that are equivalent to the following statements:
\begin{align}
x_a&=\True \label{eq:init-a} \\
x_b&=\False \label{eq:init-b} \\
x_u&\neq x_{v} \text{ if~$u$ and~$v$ have the same colour} \label{eq:colours}
\end{align}

For any bag~$X_i$ precut on some pair~$(u,v)$,
\begin{align}
x_u &\vee \neg x_v \label{eq:u-or-v}\\
\neg x_u &\Rightarrow \neg x_w\quad \forall w\in V[i] \setminus \{u\}\label{eq:propagate-u}\\
x_v &\Rightarrow x_w \quad \forall w\in V[i]\setminus \{v\}\label{eq:propagate-v}
\end{align}

For each bag $X_i=\{u\}$ of size~$1$,
\begin{align}
x_w &= x_u \quad \forall w\in V[i]\setminus\{u\}\label{eq:propagate-single}
\end{align}

For each bag $\{u,v,w\}$ of size~$3$ precut on~$(u,v)$
\begin{align}
x_w &= x_u \text{ if~$w$ is detached from~$v$} \label{eq:detach-v}\\
x_w &= x_v \text{ if~$w$ is detached from~$u$} \label{eq:detach-u}
\end{align}

\noindent
We claim that~$G$ has an $(a,b)$-partition if and only if~$\phi$ is satisfiable.

\medskip
\noindent
First suppose that~$G$ has an $(a,b)$-partition~$P$.
For each vertex~$u$ in~$G$, set $x_u=\True$ if $u\in V_1$, and~$x_u=\False$ otherwise.
As~$P$ cuts $(a,b)$, Statements~\eqref{eq:init-a} and~\eqref{eq:init-b} are satisfied.
As~$P$ is colourful, Statement~\eqref{eq:colours} is satisfied.

Consider a bag~$X_i$ precut on~$(u,v)$.
Two cases are possible: either~$P$ cuts~$(u,v)$, or~$X$ is respected by~$P$.
In the first case we have $x_u=\True$ and $x_v=\False$, which is enough to satisfy Statements~\eqref{eq:u-or-v}, \eqref{eq:propagate-u} and~\eqref{eq:propagate-v}.
In the second case, $x_u=x_v$ (which satisfies Statement~\eqref{eq:u-or-v}), and by Claim~\ref{lem:join-hereditary}, for each $j$ in~$T[i]$, $X_j$ is respected by~$P$.
Thus all vertices in~$V[i]$ are in the same subset~$V_1$ or~$V_2$, hence they have the same value of~$x_w$, which satisfies Statements~\eqref{eq:propagate-u} and~\eqref{eq:propagate-v}.

For Statement~\eqref{eq:propagate-single}, note that each bag~$X_i$ of size~$1$, and, by Claim~\ref{lem:join-hereditary}, every bag~$X_j$ such that~$j$ is a descendent of~$i$ is always respected, so again all values of~$x_w$ for $w\in V[i]$ are identical.
Finally, Statements~\eqref{eq:detach-v} and~\eqref{eq:detach-u} follow from Claim~\ref{lem:detached}.

\medskip
\noindent
Now suppose that~$\phi$ is satisfiable.
We construct two disjoint vertex sets~$V_1$ and~$V_2$ by putting every vertex~$u$ with $x_u=\True$ in~$V_1$ and every vertex with $x_u=\False$ in~$V_2$.
We claim that~$(V_1,V_2)$ is an $(a,b)$-partition of~$G$. By Statements~\eqref{eq:init-a} and~\eqref{eq:init-b}, $a \in V_1$ and $b \in V_2$.
By Statement~\eqref{eq:colours}, two vertices of the same colour cannot belong to the same part, hence it remains to prove that both sets are connected.

In fact, we shall prove by induction on~$d$, that for the set of vertices belonging to bags of nodes at distance at most~$d$ from~$0$ in~$T$, the two subsets found by dividing the set according to membership of~$V_1$ or~$V_2$ each induce a connected subgraph of~$G$.
For the base case, we consider $X_0 =\{a,b\}$ and the two subsets contain a single vertex so we are done.
For the inductive case, we consider a node~$j$ at distance $d>0$ from~$0$.
Let~$i$ be the parent of~$j$ in~$T$.
It is enough to show that any vertex in $X_j \setminus X_i$ is in the same component of~$G[V_1]$ or~$G[V_2]$ as a vertex of~$X_i$ assigned to the same set ($V_1$ or~$V_2$).
Let us first assume that~$X_j$ is in the head subtree.
If $|X_j|=2$, then $X_j \subsetneq X_i$ so we may assume that $X_j=\{u,v,w\}$ and also that $X_i =\{u,v\}$, and~$X_i$ and~$X_j$ are both precut on~$(u,v)$ (using Claims~\ref{lem:precut-X2-X3} and~\ref{lem:head-subtree-is-precut}).
We distinguish three cases.

\case{\label{case:one}$x_u=x_v=\True$.}\\
Considering Statement~\eqref{eq:propagate-v} for bag~$X_i$, we have $x_w=\True$, as well as $x_{w'}=\True$ for every vertex $w'\in V[i]\setminus X_i$.
Hence $V[i]\subseteq V_1$.
By Property~\ref{property:three}, $G[V[X]]$ is connected, so~$w$ is in the same component of~$G[V_1]$ as~$v$.

\case{\label{case:two}$x_u=x_v=\False$.}\\
Symmetrically to Case~\ref{case:one}, we can show that~$w$ is in the same component of~$G[V_2]$ as~$u$ .

\case{\label{case:three}$x_u\neq x_v$.}\\
By Statement~\eqref{eq:u-or-v}, in this case we have $x_u=\True$ and $x_v=\False$.
Suppose that $x_w=\True$ (the case for $x_w=\False$ again follows symmetrically).
Since $x_w\neq x_v$, by Statement~\eqref{eq:detach-u}, $w$ must be attached to~$u$.
If there is an edge from~$u$ to~$w$, then we are done.
If there is no such edge, then~${\cal X}$ has a bag $X_k=\{u,w\}$, and~$k$ must be a child of~$j$.
As~$k$ is in the head subtree, $X_k$ is precut on~$(u,w)$ or~$(w,u)$ by Claim~\ref{lem:head-subtree-is-precut}.
By Statements~\eqref{eq:propagate-u} and~\eqref{eq:propagate-v}, since $x_u=x_w=\True$, $x_{w'}=\True$ for every $w'\in V[k]\setminus \{u,w\}$ and thus $V[k]\subseteq V_1$.
Since~$G[V[k]]$ is connected and contains both~$u$ and~$w$, we are done.

\medskip
\noindent
Now we consider the case that~$X_j$ is not in the head subtree.
Thus~$X_j$ is in a subtree that has at its root a bag containing a single vertex~$u$ and, by Statement~\eqref{eq:propagate-single}, every vertex in the bags of the subtree are in the same subset~$V_1$ or~$V_2$ as~$u$ and these vertices together induce a connected subgraph.
Thus, as~$u$ is also in the parent of the root of the subtree (since~$G$ is connected, the parent is not empty), we are done.

\medskip
\noindent
Thus we have a polynomial-time algorithm to decide {\sc $2$-Colourful Partition} on instance~$(G,c)$.
First we compute in polynomial time a tree-decomposition~$(T,{\cal X})$ of~$G$ with Properties~\ref{property:one}, \ref{property:two} and~\ref{property:three}.
Then for every pair of adjacent vertices~$a$ and~$b$ in~$G$, we check whether there is a colourful $(a,b)$-partition of size~$2$ in polynomial time using the corresponding {\sc $2$-Satisfiability} formula~$\phi$.
\end{proof}

\section{Parameterized Complexity}\label{s-fpt}

The number of colourful components, maximum degree, number of colours, colour-multiplicity, path-width, and treewidth are all natural parameters for {\sc Colourful Partition}.
However, as a consequence of our results of the previous section, the only sensible combination of these to consider as a parameter is
the treewidth (or path-width) plus the number of colours. This can be seen as follows.

We first observe that if an instance of {\sc $k$-Colourful Partition} has colour-multiplicity more than~$k$, then it is a no-instance.
Therefore, if {\sc $2$-Colourful Partition} is \NP-complete for some class~${\cal C}$ of coloured graphs, then it must be \NP-complete for the subclass consisting of all coloured graphs in~${\cal C}$ with colour-multiplicity at most~$2$.
Now let~$X$ be a combination of the above six parameters.
Theorem~\ref{t-pathwidth}, combined with the above observation, implies that {\sc Colourful Partition} is para-\NP-complete when parameterized by~$X$ if~$X$ is any subset (including the whole set) of the number of colourful components, maximum degree, colour-multiplicity, path-width and treewidth.
Hence, we may assume that~$X$ contains the number of colours.
If~$X$ also contains the number of colourful components or the colour-multiplicity, then {\sc Colourful Partition} is trivially \FPT, as the size of the input is bounded.
Assume that~$X$ contains neither the number of colourful components nor the colour-multiplicity.
If~$X$ contains neither the path-width nor the treewidth, then {\sc Colourful Partition} is para-\NP-complete due to Theorem~\ref{t-33}.
Hence, we are indeed left to consider the case where $X$ contains the number of colours and the path-width or treewidth. 
For this case we are able to prove the following result.

\begin{theorem}\label{t-treewidth-numberofcolours}
{\sc Colourful Partition} and {\sc Colourful Components} are \FPT\ when parameterized by the number of colours plus the treewidth.
\end{theorem}

\begin{proof}
Let~$(G,c)$ be a coloured graph and let~$C$ be the set of colours used on~$G$.
We may use the algorithm of~\cite{Bo96} to obtain a tree decomposition of the input graph~$G$ in \FPT\ time.
We then convert it to a nice tree decomposition~$T$ using the aforementioned result of Kloks~\cite{Kl94}.

Let~$r$ denote the root of~$T$.
For any node $i\in V(T)$, let~$T_i$ denote the subtree of~$T$ induced by~$i$ and its descendants and let $G_i=G[\bigcup_{j\in V(T_i)}X_j]$.
For both problems, we apply a dynamic programming algorithm over~$(T,{\cal X})$.

\medskip
\noindent
First, we describe the tables that are constructed for the nodes of~$T$ when solving the {\sc Colourful Partition} problem.
Let $i\in V(T)$.
We define~${\bf table}_i$ as a partial function whose inputs are ordered pairs~$(P,\rho)$ where
\begin{itemize}
\item $P$ is a partition of~$X_i$, and
\item $\rho: P \to {\cal{P}}(C)$ is a function assigning a set of colours to each set in~$P$.
\end{itemize}

\noindent
Let~$Q$ be a partition of the vertices of~$G_i$ into colourful components.
We say that~$Q$ {\em induces} the partition~$P$ on~$X_i$ and the function $\rho: P \to {\cal{P}}(C)$ if both the following conditions hold:
\begin{itemize}
\item Two elements of~$X_i$ are in the same set in~$P$ if and only if they are in the same set in~$Q$.
\item For~$D \in P$, let $D' \in Q$ be the colourful component that contains the vertices of~$D$.
Then $\rho(D) = \{c(v) \;|\; v \in D' \setminus D\}$.
\end{itemize}

\noindent
For a pair $(P,\rho)$, the value of ${\bf table}_i(P,\rho)$ will be the minimum possible number of colourful components in a colourful partition~$Q$ of~$G_i$ among all such partitions~$Q$ that induce $(P,\rho)$; if no such partition exists then ${\bf table}_i(P,\rho)$ is void.

Recall that $X_r=\emptyset$.
Therefore ${\bf table}_r(\emptyset,\varnothing)$ is the minimum number of colourful components for which~$G$ has a colourful partition.

\newcommand{\improves}{\leftarrowtail}

Now we explain how we construct~${\bf table}_i$ for each $i\in V(T)$.
In what follows we write ${\bf table}_i(P,\rho) \improves t$ to refer to the following procedure:
If ${\bf table}_i(P,\rho)$ is undefined, set it to be equal to~$t$.
If ${\bf table}_i(P,\rho) = t'>t$ then set ${\bf table}_i(P,\rho)=t$, otherwise, leave ${\bf table}_i(P,\rho)$ unchanged.

If~$i$ is a {\em leaf} node, ${\bf table}_i$ is constructed in a straightforward way because $X_i=\emptyset$, so we set ${\bf table}_i(\emptyset,\varnothing) \improves 0$ and all other entries in ${\bf table}_i$ remain void.

If~$i$ is an {\em introduce} node, let~$j$ be the unique child node of~$i$ and let $\{v\}=X_i \setminus X_j$.
For every pair $(P,\rho)$ such that ${\bf table}_j(P,\rho)$ is not void, we do the following.
Consider every possible subset of sets $R \subseteq P$ such that~$v$ has at least one neighbour in every set of~$R$ and $Y_v=\{v\}\cup \bigcup R$ is colourful (note that~$R$ may be empty).
If the sets in $\{\rho(Y) \;|\; Y \in R\} \cup \{\{c(u)\;|\; u\in Y_v\}\}$ are not pairwise disjoint, then disregard this choice of~$R$. 
Otherwise, let $P'=P\setminus R \cup \{Y_v\}$ and $\rho':P' \to {\cal P}(C)$ such that $\rho'(Y)=\rho(Y)$ for $Y \in P \setminus R$ and $\rho'(Y_v)=\bigcup_{Y \in R}\rho(Y)$.
Apply ${\bf table}_i(P',\rho')\improves {\bf table}_j(P,\rho)+1-|R|$.

If~$i$ is a {\em forget} node, let~$j$ be the unique child node of~$i$ and let $\{v\}=X_j \setminus X_i$.
For every pair $(P,\rho)$ such that ${\bf table}_j(P,\rho)$ is not void, let~$P'$ be the partition~$P$ restricted to~$X_i$ (that is, if $\{v\} \in P$ then delete this set from~$P$ and otherwise, remove~$v$ from the set in~$P$ that contains it) and let $\rho':P' \to {\cal P}(C)$ be the function that takes the value $\rho'(Y)=\rho(Y)$ if $Y \in P\cap P'$ and $\rho'(Y)=\rho(Y\cup\{c(v)\})$ otherwise.
Apply ${\bf table}_i(P',\rho')\improves {\bf table}_j(P,\rho)$.

If~$i$ is a {\em join} node, let~$j$ and~$k$ be the two child nodes of~$i$.
For every pair of pairs $(P,\rho)$, $(P',\rho')$ such that ${\bf table}_j(P,\rho)$ and ${\bf table}_k(P',\rho')$ are not void, we do the following.
Recall that~$P$ and~$P'$ are partitions of $X_i=X_j=X_k$.
We construct a partition~$P''$ and a function~$\rho''$ as follows:
Start by setting $P''=P$.
If two sets $P_1,P_2 \in P''$ contain vertices $v_1,v_2$, respectively, such that~$v_1$ and~$v_2$ are in the same set of~$P'$, then replace~$P_1$ and~$P_2$ by $P_1\cup P_2$ in~$P''$.
Repeat this process exhaustively and note that the resulting partition~$P''$ is a coarsening of both~$P$ and~$P'$.
Consider an element $Y \in P''$.
If~$Y$ is not colourful, then we discard this pair of pairs $(P,\rho)$, $(P',\rho')$, and consider the next pair.
Since~$P''$ is a coarsening of~$P$ and~$P'$, there must exist sets $P_1,\ldots,P_a \in P$ and $P_1',\ldots,P_b' \in P'$ such that $P_1\cup\cdots\cup P_a=P_1'\cup\cdots\cup P_b'=Y$.
Let~$c(Y)$ denote the set of colours used on~$Y$.
If the sets $c(Y),\rho(P_1),\ldots,\rho(P_a),\rho'(P_1'),\ldots,\rho'(P_b')$ are not pairwise disjoint, then we discard this pair of pairs $(P,\rho)$, $(P',\rho')$, and consider the next pair.
Otherwise, set $\rho''(Y)=\rho(P_1)\cup\cdots\cup\rho(P_a)\cup\rho'(P_1')\cup\cdots\cup\rho'(P_b')$.
Apply ${\bf table}_i(P'',\rho'')\improves {\bf table}_j(P,\rho)+{\bf table}_k(P',\rho')-|P|-|P'|+|P''|$.

It is easy to verify that the above procedure will complete ${\bf table}_i(P,\rho)$ correctly and we can thus obtain the size of an optimal solution to {\sc Colourful Partition}.
Note that for $i \in T$, the set~$X_i$ contains at most $\tw(G)+1$ elements and so the number of partitions~$P$ of~$X_i$ that we need to consider is bounded above by $(\tw(G)+1)^{\tw(G)+1}$, that is, a function of the treewidth of~$G$.
For each such partition~$P$, there are at most $2^{|C||P|}$ functions $\rho: P \to {\cal P}(C)$ that we need to consider.
Thus, for each~$i$, the number of pairs $(P,\rho)$ that need to be considered is bounded by a function of~$\tw(G)$ and~$|C|$.
Thus, for every~$i$, ${\bf table}_i$ can be completed in \FPT-time.
Since the number of nodes in~$T$ is~$O(n)$, it follows that {\sc Colourful Partition} can be solved in \FPT-time parameterized by $\tw(G)+|C|$.
This completes the proof for the {\sc Colourful Partition} problem.

\medskip
\noindent
We now describe the tables that are constructed for the nodes of~$T$ when solving the {\sc Colourful Components} problem.
Let $i\in V(T)$.
We define~${\bf table}'_i$ as a partial function whose inputs are ordered pairs~$(P,\rho)$ where
\begin{itemize}
\item $P$ is a partition of~$X_i$, and
\item $\rho: P \to {\cal{P}}(C)$ is a function assigning a set of colours to each set in~$P$.
\end{itemize}

\noindent
Note that the {\sc Colourful Components} problem is equivalent to finding the minimum number of edges that must be deleted from the input graph~$G$ such that the vertex set of the resulting graph can be partitioned into sets $V_1,\ldots,V_k$ (for arbitrary~$k$) such that each set~$V_i$ is colourful.
Note that in this case the number of sets $V_1,\ldots,V_k$ does not matter for our purposes and we do not insist that~$G[V_i]$ is connected.
However, we do insist that for distinct $i,j \in \{1,\ldots,k\}$, all edges with one end-vertex in~$V_i$ and the other in~$V_j$ must be in the set of deleted edges.
Furthermore, note that in any optimal solution no edge will be deleted with both end-vertices in the same set~$V_i$.
Thus the {\sc Colourful Components} problem is equivalent to finding a partition of~$V(G)$ into colourful sets $V_1,\ldots,V_k$ that minimizes the number of edges whose end-vertices are in different sets of the partition.
We call a partition of~$V(G)$ into colourful sets $V_1,\ldots,V_k$ a {\em colourful set partition} (note that this is different from the definition of colourful partition, for which we insist that each partition set induces a connected graph in~$G$).

Let~$Q$ be a colourful set partition of the vertices of~$G_i$.
Similarly to the {\sc Colourful Partition} case, we say that~$Q$ {\em induces} the partition~$P$ on~$X_i$ and the function $\rho: P \to {\cal{P}}(C)$ if both the following conditions hold:
\begin{itemize}
\item Two elements of~$X_i$ are in the same set in~$P$ if and only if they are in the same set in~$Q$.
\item For~$D \in P$, let $D' \in Q$ be the set that contains the vertices of~$D$.
Then $\rho(D) = \{c(v) \;|\; v \in D' \setminus D\}$.
\end{itemize}

\noindent
For a pair $(P,\rho)$, the value of ${\bf table}'_i(P,\rho)$ will be the minimum possible number of edges that need to be deleted in a colourful set partition~$Q$ of~$G_i$ among all such partitions~$Q$ that induce $(P,\rho)$; if no such partition exists then ${\bf table}'_i(P,\rho)$ is void.

Recall that $X_r=\emptyset$.
Therefore ${\bf table}'_r(\emptyset,\varnothing)$ is the minimum number edges that need to be deleted from~$G$ to obtain a colourful set partition.

Now we explain how we construct~${\bf table}'_i$ for each $i\in V(T)$.
Similarly to the {\sc Colourful Partition} case, in what follows we write ${\bf table}'_i(P,\rho) \improves t$ to refer to the following procedure:
If ${\bf table}'_i(P,\rho)$ is undefined, set it to be equal to~$t$.
If ${\bf table}'_i(P,\rho) = t'>t$ then set ${\bf table}'_i(P,\rho)=t$, otherwise, leave ${\bf table}'_i(P,\rho)$ unchanged.

If~$i$ is a {\em leaf} node, ${\bf table}'_i$ is constructed in a straightforward way because $X_i=\emptyset$, so we set ${\bf table}'_i(\emptyset,\varnothing) \improves 0$ and all other entries in ${\bf table}'_i$ remain void.

If~$i$ is an {\em introduce} node, let~$j$ be the unique child node of~$i$ and let $\{v\}=X_i \setminus X_j$.
For every pair $(P,\rho)$ such that ${\bf table}'_j(P,\rho)$ is not void, we do the following.
We choose each $R \in P$ in turn and let~$n_R$ be the number of neighbours~$v$ has in $X_i \setminus R$.
If $c(v) \in \rho(R)$ or there is a vertex in~$R$ with the same colour as~$v$, then we discard this choice of~$R$ and move on to the next one.
Otherwise, let $P'=P \setminus \{R\} \cup \{R \cup \{v\}\}$ and $\rho':P' \to {\cal P}(C)$ such that $\rho'(Y)=\rho(Y)$ for $Y \in P \setminus R$ and $\rho'(R \cup \{v\}\})=\rho(R)$.
Apply ${\bf table}'_i(P',\rho')\improves {\bf table}'_j(P,\rho)+n_R$, then move onto the next choice of~$R$.
Finally, let $n_\emptyset$ be the number of neighbours that~$v$ has in $X_i \setminus \{v\}$, let $P'=P \cup \{\{v\}\}$ and let $\rho':P' \to {\cal P}(C)$ such that $\rho'(Y)=\rho(Y)$ for $Y \in P$ and $\rho'(\{v\})=\emptyset$. Apply ${\bf table}'_i(P',\rho')\improves {\bf table}'_j(P,\rho)+n_\emptyset$.

If~$i$ is a {\em forget} node, let~$j$ be the unique child node of~$i$ and let $\{v\}=X_j \setminus X_i$.
For every pair $(P,\rho)$ such that ${\bf table}'_j(P,\rho)$ is not void, let~$P'$ be the partition~$P$ restricted to~$X_i$ (that is, if $\{v\} \in P$ then delete this set from~$P$ and otherwise, remove~$v$ from the set in~$P$ that contains it) and let $\rho':P' \to {\cal P}(C)$ be the function that takes the value $\rho'(Y)=\rho(Y)$ if $Y \in P\cap P'$ and $\rho'(Y)=\rho(Y\cup\{c(v)\})$ otherwise.
Apply ${\bf table}'_i(P',\rho')\improves {\bf table}'_j(P,\rho)$.

If~$i$ is an {\em join} node, let~$j$ and~$k$ be the two child nodes of~$i$.
For every pair of pairs $(P,\rho)$, $(P,\rho')$ such that ${\bf table}'_j(P,\rho)$ and ${\bf table}'_k(P,\rho')$ are not void, we do the following (note that the partition~$P$ in each pair is the same).
If $\rho(R) \cap \rho'(R) = \emptyset$ for all $R \in P$ then let~$e_P$ be the number of edges in~$G[X_i]$ whose end-vertices are in distinct sets of~$P$ and apply ${\bf table}'_i(P,\rho)={\bf table}'_j(P,\rho)+{\bf table}'_k(P,\rho')-e_P$.

Similarly to the case for {\sc Colourful Partition}, it is easy to verify that the above procedure will complete ${\bf table}'_i(P,\rho)$ correctly and we can thus obtain the size of an optimal solution to {\sc Colourful Components}.
Furthermore, it is easy to verify that the procedure also runs in \FPT-time.
This completes the proof.
\end{proof}

We now prove two \FPT\ results for two different parameters.
Our proof for the next result uses similar arguments to the proof sketch of Theorem~\ref{t-vertexcover_c} given in~\cite{Mi18}.
However, the details of both proofs are different, as optimal solutions for {\sc Colourful Partition} do not necessarily translate into optimal solutions for {\sc Colourful Components}.
This holds even if the instance has a vertex cover of size~$2$, as we showed in Example~\ref{e-intro} (the vertices $w,w'$ form a vertex cover in~$G$).

\begin{theorem}\label{t-vertexcover}
{\sc Colourful Partition} is \FPT\ when parameterized by vertex cover number.
\end{theorem}

\begin{proof}
In fact we will show the result for the optimization version of {\sc Colourful Partition}.
Let~$(G,c)$ be a coloured graph.
We will prove that we can find the size of a {\em minimum} colourful partition (one with smallest size) in \FPT\ time.
By a simple greedy argument, we can find a vertex cover~$S$ of~$G$ that contains at most $2\vc(G)$ vertices.
It is therefore sufficient to show that {\sc Colourful Partition} is \FPT\ when parameterized by~$|S|$.
If two vertices of~$G$ have the same colour then they will always be in different colourful components of~$G$.
Thus if two vertices with the same colour are adjacent, we can delete the edge that joins them, that is, we may assume that~$c$ is a proper colouring of~$G$ (note that deleting edges from~$G$ maintains the property that~$S$ is a vertex cover).
Since~$S$ is a vertex cover, $T=V(G)\setminus S$ is an independent set.
Let~$C$ be the set of colours used on~$G$.
We let $s=|S|$ and for a set $S'\subseteq S$, we let~$T_i(S')$ be the set of vertices with colour~$i$ whose neighbourhood in~$S$ is~$S'$, that is, for all $u\in T_i(S')$, we have that $N_S(u)=S'$ and $c(u)=i$.

\ourrule{\label{rule:one}If there is a colour~$i\in C$ and a set~$S'\subseteq S$ such that $|T_i(S')|\geq s+1$, then delete $|T_i(S')|-s$ (arbitrary) vertices of~$T_i(S')$ from~$G$.}

\medskip
\noindent
We claim that we can safely apply Rule~\ref{rule:one}.
In order to see this, consider any colourful partition $(V_1,\ldots,V_k)$ of $(G,c)$.
The number of colourful components in the solution containing at least one vertex of~$S$ is at most~$s$.
Hence at most~$s$ vertices of~$T_i(S')$ can occur in these components.
All other vertices of~$T_i(S')$ will be in $1$-vertex components of the solution.
Since the vertices of~$T_i(S')$ have the same neighbourhood, we may choose the latter set of vertices arbitrarily.
Then, given a minimum colourful partition for the resulting coloured graph~$(G',c')$, we obtain a minimum colourful partition for~$(G,c)$ by restoring the deleted vertices into $1$-vertex components.
This proves the claim.\dia

\medskip
\noindent
We apply Rule~\ref{rule:one} exhaustively.
For convenience we again denote the resulting instance by $(G,c)$ and let $T=V(G)\setminus S$.
Note that now $|T_i(S')| \leq s$ for every $i\in C$ and every $S'\subseteq S$.
Consequently, the number of vertices of~$T$ with colour~$i$ is at most~$s2^s$.
Note that this means that $|T|\leq |C|s2^s$ and thus the total number of vertices in~$G$ is at most $s+|C|s2^s$, which gives us an \FPT-algorithm in $s+|C|$.
Hence in order to prove our result it remains to bound the size of~$C$ by a function of~$s$.

We let~$C_T$ denote the set of colours that appear on vertices of~$T$ but not on vertices of~$S$.
For two colours $i,j \in C_T$, if $T_i(S')=T_j(S')$ holds for all $S' \subseteq S$, then we say that~$i$ and~$j$ are {\em clones} and note that these colours are interchangeable.
For a colourful partition $P=(V_1,\ldots,V_k)$ for~$G$, we let $P_S=(V_1\cap S,\ldots,V_k\cap S)$ be the partition of~$S$ {\em induced by}~$P$ (note that in this case we allow some of the blocks to be empty).

We consider each partition~$Q$ of~$S$.
If a block of~$Q$ is not colourful, then we discard~$Q$.
Otherwise we determine a minimum colourful partition~$P$ for~$G$ with $P_S=Q$; note that such a partition~$P$ may not exist, as it may not be possible to make the blocks of~$Q$ connected (by using vertices~$T$ in addition to edges both of whose endpoints lie in~$S$).
Finally, we will choose the colourful partition for~$G$ that has minimum size overall.

Let~$Q$ be a partition of~$S$ in which each block is colourful.

\ourrule{\label{rule:two}If there are~$s$ distinct colours $i_1,\ldots,i_s$ in~$C_T$ that are pairwise clones, then delete all vertices with colour~$i_s$ from~$G$.}

\medskip
\noindent
We claim that can safely apply Rule~\ref{rule:two} if we only consider colourful partitions~$P$ that induce~$Q$ on~$S$.
We say that a colour~$i\in C_T$ is {\em redundant} for a colourful partition~$P$ if after deleting all vertices of~$T$ with colour~$i$, the resulting colourful partition~$P'$ induces the same partition on~$S$ as~$P$ does, that is $P_S'=P_S$.
Suppose we have~$s$ distinct colours $i_1,\ldots,i_s\in C_T$ that are pairwise clones.
Any induced partition~$Q$ on~$S$ requires using at most~$s-\nobreak 1$ vertices of~$T$ to connect the vertices in the same block of~$Q$.
Hence, there must be at least one colour that is redundant for~$P$.
As the colours $i_1,\ldots,i_s$ are indistinguishable in~$(G,c)$ it does not matter which colour we choose to delete, so we may assume that~$i_s$ is redundant.

Let~$(G',c')$ be the resulting coloured graph after removing all vertices with colour~$i_s$ from~$G$.
Let~$P'$ be a minimum colourful partition for~$(G',c')$ with $P'_S=Q$.
Then we compute a minimum colourful partition~$P$ of~$G$ with $P_S=Q$ as follows.
We construct an auxiliary bipartite graph~$F$ with partition classes~$X$ and~$Y$.
The vertices of~$X$ represent the vertices of~$T$ with colour~$i_s$, and the vertices of~$Y$ represent colourful components of~$P$ that contain a vertex of~$S$.
We add an edge between two vertices $x\in X$ and $y\in Y$ if and only if the vertex of~$T$ corresponding to~$x$ is adjacent to at least one vertex of the component of~$P$ corresponding to~$y$.
We then compute a maximum matching~$M$ in~$F$, which prescribes how the vertices with colour~$i_s$ must be added to~$P'$ to obtain~$P$.
Note that $P_S=P'_S=Q$.
As~$M$ is a matching, at most one vertex with colour~$i_s$ is added to the components of~$P'$ that contain a vertex of~$S$.
Hence, $P$ is colourful.
As~$M$ is maximum, the number of components that consist of isolated vertices coloured~$i_s$ is minimized.
Hence, $P$ is a minimum colourful partition of~$(G,c)$ with $P_S=Q$.\dia

\medskip
\noindent
We now apply Rule~\ref{rule:two} exhaustively, 
and again
call the resulting graph~$G$ and define~$S$ and~$T$ as before.

\setcounter{ctrclaim}{0}
\medskip
\noindent
\clm{\label{clm:one}$|C_T|\leq (s-1)(s+1)^{2^s}$.}

\medskip
\noindent
We prove Claim~\ref{clm:one} as follows.
First note that $0\leq |T_i(S')| \leq s$ for every $i\in C$ and every $S'\subseteq S$ by Rule~\ref{rule:one}.
Thus for every~$S' \subseteq S$, there are $s+1$ possible values of $|T_i(S')|$.
Since~$S$ has~$2^s$ subsets~$S'$, for any colour~$i$ there are at most $(s+1)^{2^s}$ possible mappings $|T_i(S')|: {\cal P}(S) \to \{0,\ldots,s\}$.
If two colours $i$, $j$ have the same mapping, they are clones.
By Rule~\ref{rule:two}, at most $s-1$ colours can be pairwise clones.
Therefore $|C_T|\leq (s-1)(s+1)^{2^s}$.\dia

\medskip
\noindent
We continue as follows.
The number of different colours used on vertices of~$S$ is at most~$s$.
Hence $C\setminus C_T$ has size at most~$s$.
Recall that for every colour $i\in C$, the number of vertices of~$T$ with colour~$i$ is at most~$s2^s$.
We combine these two facts with  Claim~\ref{clm:one}.  Then
\[\begin{array}{lcl}
|V| &= &|S| + |T|\\
&\leq &s + |C|s2^s\\
&= &s+|C\setminus C_T|s2^s + |C_T|s2^s\\
&\leq &s+s^22^s+ (s-1)(s+1)^{2^s}s2^s,
\end{array}\]
which means that 
by brute force we can compute a minimum colourful partition~$P$ for~$G$ subject to the restriction that $P_S=Q$ in~$f(s)$ time for some function~$f$ that only depends on~$s$.

The correctness of our \FPT-algorithm follows from the above description.
It remains to analyze the running time.
Applying Rule~\ref{rule:one} exhaustively takes $O(2^s|C|)=O(2^sn)$ time, as the number of different subsets $S'\subseteq S$ is~$2^s$.
We then branch into at most~$s^s$ directions by considering every partition of~$S$.
Applying Rule~\ref{rule:two} exhaustively takes $O(2^sn^2)$ time, as for each colour~$i\in C_T$ we first calculate the values of $|T_i(S')|$ for every $S' \subseteq S$, which can be done in $O(2^sn)$ time.
Doing this for every colour takes a total of $O(2^sn^2)$ time and partitioning the colours into sets that are clones can be done in $O(2^{2s}n^2)$ time, and deleting colours can be done in~$O(sn)$ time.
As every auxiliary graph~$F$ has at most~$n$ vertices, we can compute a maximum matching in~$F$ in~$O(n^{\frac{5}{2}})$ time by using the Hopcroft-Karp algorithm~\cite{HK73}.
Finally, translating a minimum solution into a minimum solution for the graph in which we restored the vertices we removed due to exhaustive application of Rules~\ref{rule:one} and~\ref{rule:two} takes~$O(n)$ time.
This means that the total running time is $O(2^sn)+s^s(O(2^{2s}n^2)+O(s)+O(n^{\frac{5}{2}})+f(s)+O(n))=f'(s)O(n^\frac{5}{2})$ for some function~$f'$ that only depends on~$s$, as desired.
\end{proof}

\noindent
For our final result we need to
introduce a problem used by Robertson and Seymour~\cite{RS95} in their graph minor project to prove their algorithmic result on recognizing graphs that contain some fixed graph~$H$ as a minor.
They proved that
this problem
is cubic-time solvable as long as $Z_1\cup \cdots \cup Z_r$ has constant size.

\problemdef{Disjoint Connected Subgraphs}{a graph~$G$ and~$r$ pairwise disjoint subsets $Z_1,\ldots,Z_r$ of~$V(G)$ for some $r\geq 1$.}{Is it possible to partition $V(G)\setminus (Z_1\cup \cdots \cup Z_r)$ into sets $S_1,\ldots,S_r$ such that every $S_i\cup Z_i$ induces a connected subgraph of~$G$?}

\noindent
We are now ready to prove our final result.

\begin{theorem}\label{t-nonunique}
When parameterized by the number of non-uniquely coloured vertices, {\sc Colourful Components} is para-\NP-complete, but {\sc Colourful Partition} is \FPT.
\end{theorem}

\begin{proof}
The {\sc Multiterminal Cut} problem is to test for a graph~$G$, integer~$p$ and terminal set~$S$, if there is a set~$E'$ with $|E'|\leq p$ such that every terminal in~$S$ is in a different component of $G-E'$.
This problem is \NP-complete even if $|S|=3$~\cite{DJPSY94}.
To prove the first part, give each of the three vertices in~$S$ colour~$1$ and the vertices in $G-S$ colours $2,\ldots,|V|-2$.

To prove the second part, let~$(G,c)$ be a coloured graph and~$k$ be an integer be an instance of {\sc Colourful Partition}.
We 
may 
assume without loss of generality that~$G$ is connected.
Let~$Q$ with $|Q|=q$ be the set of non-uniquely coloured vertices.
If $k\geq q$, then 
we 
place each of the~$q$ vertices of~$Q$ in a separate component and assign the uniquely coloured vertices to components in an arbitrary way subject to maintaining connectivity of the~$q$ components.
This yields a colourful partition of~$(G,c)$ of size at most~$k$.
Now assume that $k\leq q-1$.
We consider every possible partition of~$Q$ into~$k$ sets $Z_1,\ldots,Z_k$, where
some of the sets~$Z_i$ may be empty.
It remains to solve {\sc Disjoint Connected Subgraphs} on the input $(G,Z_1,\ldots,Z_k)$.
Note that $|Z_1|+\cdots+|Z_k|$ has size~$q$.
Hence, by the 
aforementioned 
result of Robertson and Seymour~\cite{RS95}, solving {\sc Disjoint Connected Subgraphs} takes cubic time.
As there are~$O(q^q)$ partitions to consider, the result follows.
\end{proof}
As an immediate consequence of Theorem~\ref{t-nonunique} we obtain the following result.

\begin{corollary}\label{c-reverse}
There exists a family of instances on which {\sc Colourful Components} and {\sc Colourful Partition} have different parameterized complexities.
\end{corollary}

\section{Conclusions}\label{s-con}

We conclude our paper with some open problems.
We showed that {\sc Colourful Partition} and {\sc Colourful Components} are \NP-complete for coloured trees of maximum degree at most~$6$ (and colour-multiplicity~$2$).
What is their complexity for coloured trees of maximum degree~$d$ for $3\leq d\leq 5$?
{\sc Colourful Components} is known to be \NP-complete for $3$-coloured graphs of maximum degree~$6$ (Theorem~\ref{t-3colours-degree6});
we also ask if one can prove a result analogous to Theorem~\ref{t-33}: what is its complexity for coloured graphs of maximum degree~$3$?

We also proved that {\sc $2$-Colourful Partition} is \NP-complete for coloured (planar bipartite) graphs of path-width~$3$ (and maximum degree~$3$), but polynomial-time solvable for coloured graphs of treewidth~$2$.
We believe that the latter result can be extended to {\sc $k$-Colourful Partition} $(k\geq 3$), but leave this for future research.
A more interesting question is whether the problem is \FPT\
on graphs of
treewidth~$2$ when parameterized by $k$.

\bibliography{mybib-fsttcs}

\end{document}